\documentclass[conference]{IEEEtran}
\IEEEoverridecommandlockouts
\usepackage{etoolbox}
\makeatletter
\patchcmd{\@makecaption}
  {\scshape}
  {}
  {}
  {}
\makeatother
\usepackage[noblocks]{authblk}
\usepackage{algorithm}
\usepackage{algpseudocode}
\usepackage{qcircuit}
\usepackage{physics}
\usepackage{tikz}
\usepackage{amsmath}
\usepackage{amsthm}
\usepackage{amssymb}
\newtheorem{definition}{Definition}
\newtheorem{lemma}{Lemma}
\newtheorem{theorem}{Theorem}
\newcommand*\circled[1]{\tikz[baseline=(char.base)]{
            \node[shape=circle,draw,inner sep=1pt] (char) {#1};}}

\AtBeginDocument{%
  \providecommand\BibTeX{{%
    \normalfont B\kern-0.5em{\scshape i\kern-0.25em b}\kern-0.8em\TeX}}}

\begin{document}

\title{Shallow Quantum Circuit Implementation of Symmetric Functions with Limited Ancillary Qubits}

\author[1,2]{Wei Zi}
\author[1,2]{Junhong Nie}
\author[1,2, *]{Xiaoming Sun}
\affil[1]{State Key Lab of Processors, Institute of Computing Technology, Chinese Academy of Sciences, Beijing 100190, China}
\affil[2]{School of Computer Science and Technology, University of Chinese Academy of Sciences, Beijing 100049, China}
\affil[*]{sunxiaoming@ict.ac.cn}

\maketitle
\begin{abstract}
    In quantum computation, optimizing depth and number of ancillary qubits in quantum circuits is crucial due to constraints imposed by current quantum devices. This paper presents an innovative approach to implementing arbitrary symmetric Boolean functions using poly-logarithmic depth quantum circuits with logarithmic number of ancillary qubits. Symmetric functions are those whose outputs rely solely on the Hamming weight of the inputs. These functions find applications across diverse domains, including quantum machine learning, arithmetic circuit synthesis, and quantum algorithm design (e.g., Grover's algorithm). Moreover, by fully leveraging the potential of qutrits (an additional energy level), the ancilla count can be further reduced to 1. The key technique involves a novel poly-logarithmic depth quantum circuit designed to compute Hamming weight without the need for ancillary qubits. The quantum circuit for Hamming weight is of independent interest because of its broad applications, such as quantum memory and quantum machine learning.
\end{abstract}

\section{introduction}

Quantum computing has attracted considerable attention for its potential to surpass classical computing, exemplified by Shor's algorithm \cite{shor1994algorithms}. There have been notable advancements in hardware implementations and quantum computing software optimizations in recent years. One of the most critical research directions is the optimization of quantum circuits \cite{10.1145/3581784.3607032,10046102,10.1145/3458817.3476189,bullock2005asymptotically}. The optimization of circuit depth and ancilla count has garnered significant attention from researchers, as these two measures correspond to the time and space complexity of quantum computation, respectively \cite{maslov2021quantum,takahashi2021power,chattopadhyay2016low,ISCAqutrit,zidac2023,nie2023quantum}.

Symmetric functions \cite{canteaut2005symmetric} are Boolean functions whose outputs depend solely on the Hamming weight of their inputs. These functions have wide-ranging applications in reversible circuit synthesis \cite{yang2005majority}, quantum adder \cite{biswal2022efficient}, and quantum machine learning \cite{nashiry2017reversible}. As a special class of Boolean functions, symmetric functions can serve as oracles in quantum algorithms such as Grover's algorithm \cite{grover1996fast}. The quantum circuit implementation of symmetric functions has garnered attention from researchers \cite{10.1145/2894757,deb2013reversible}. Perkowski et al. \cite{perkowski2001regular,53f780b7a81e471894c6710887a330ac} proposed a quantum circuit with $O(n^2)$ depth to implement symmetric functions, utilizing $O(n^2)$ ancillary qubits, where $n$ is the number of input qubits. Maslov and Dmitri \cite{maslov2004dynamic,maslov2006efficient} maintained the same circuit depth while reducing the ancilla count to $O(n)$. Anupam and Anubhab \cite{chattopadhyay2016low} devised a circuit with $O(n\log n)$ depth to implement symmetric functions with $O(n)$ ancillary qubits. Takahashi and Tani \cite{takahashi2021power} were able to reduce the circuit depth to $O(\log^2 n)$. However, their approach requires $O(\log n)$ clean ancillary qubits and $O(n \log^2 n)$ borrowed ancillary qubits. Maslov et al. \cite{maslov2021quantum} implement symmetric functions without ancillary qubits but at the expense of $O(n^2)$ circuit depth.

In this paper, we introduce a novel approach to implementing arbitrary symmetric functions using quantum circuits of depth $O(\log^2 n)$, requiring only $\lceil \log (n+1) \rceil$ clean ancillary qubits. We are the first to demonstrate that symmetric functions can be implemented with lower circuit depth and ancilla count simultaneously. We provide a comparison of these results in Table~\ref{tab:result}. It's worth noting that any Boolean function can be transformed into a symmetric function by repeating the input \cite{chrzanowska1999logic}. In \cite{1674562}, a method is proposed to optimize the number of repetitions for symmetrization, greatly extending potential applications of our result. Our key technique is an efficient quantum circuit designed to compute the Hamming weight.

\begin{table}[h]
    \centering
    \caption{The comparison of circuit depth and ancilla count of quantum circuits implementing symmetric functions.}
    \renewcommand{\arraystretch}{1.5}
    \begin{tabular}{ccc}
    \hline \hline
         & circuit depth & ancilla count \\  \hline
       \cite{perkowski2001regular,53f780b7a81e471894c6710887a330ac}  & $O(n^2)$ & $O(n^2)$ \\ 
       \cite{maslov2004dynamic,maslov2006efficient} & $O(n^2)$ & $O(n)$ \\ 
       \cite{chattopadhyay2016low} & $O(n\log n)$ & $O(n)$ \\ 
       \cite{takahashi2021power} & $O(\log^2 n)$ & $O(n\log^2n)$ \\ 
       \cite{maslov2021quantum} & $O(n^2)$ & $0$ \\ 
       our result & $O(\log^2n)$ & $\lceil \log (n+1) \rceil$ \\ \hline \hline
    \end{tabular}
    \label{tab:result}
\end{table}

The Hamming weight of a bit-string is the count of 1's in the string, and the Hamming distance represents the number of positions at which two bit-strings of the same length differ. Quantum circuit for Hamming distance is broadly applicable in quantum memory \cite{khan2022ep}, and quantum machine learning \cite{kathuria2020implementation}, particularly in quantum nearest neighbor classification algorithms \cite{ruan2017quantum,li2022quantum,berti2022effect,li2023quantum,tian2021fake}. Indeed, any efficient quantum circuit implementation of Hamming weight implies an efficient implementation for Hamming distance simply by adding two layers of CNOT gates. Therefore, the quantum circuit implementation of Hamming weight is of independent interest. To implement a quantum circuit for computing the Hamming weight of an $n$-bit string, $\lceil \log (n+1) \rceil$ output qubits are necessary to store the results and are therefore not included in the ancilla count. In \cite{li2022quantum,li2023quantum}, the Hamming weight is computed using a quantum circuit with a depth of $n\log ^2n$ and $O(1)$ ancillary qubits. In \cite{chattopadhyay2016low,orts2024quantum}, although the circuit depth is reduced to $O(n)$, the ancilla count increases to $O(n)$. In \cite{takahashi2021power}, the circuit depth is further reduced to $O(\log ^2 n)$ but using $O(n \log ^2 n)$ borrowed ancillary qubits.

In this paper, we propose a quantum circuit with a depth of $O(\log ^2n)$ to compute the Hamming weight without the need for ancillary qubits. It achieves the best circuit depth and ancilla count simultaneously compared to previous results. Table~\ref{tab:result2} presents a comparison of results regarding Hamming weight.

\begin{table}[h]
    \centering
    \caption{The comparison of circuit depth and ancilla count of quantum circuits compute the Hamming weight.}
    \renewcommand{\arraystretch}{1.5}
    \begin{tabular}{ccc}
    \hline \hline
         & circuit depth & ancilla count \\  \hline
       \cite{li2022quantum,li2023quantum}  & $O(n\log ^2n)$ & $O(1)$ \\ 
       \cite{chattopadhyay2016low,orts2024quantum} & $O(n)$ & $O(n)$ \\ 
       \cite{takahashi2021power} & $O(\log ^2n)$ & $O(n\log ^2n)$ \\ 
       our result & $O(\log^2 n)$ & $0$ \\ \hline \hline
    \end{tabular}
    \label{tab:result2}
\end{table}

The implementation of symmetric functions consists of two steps. In the first step, we construct a quantum circuit of depth $O(\log^2 n)$ to compute the Hamming weight of the input and store the result in $ \lceil \log (n+1) \rceil$ clean ancillary qubits. This turns a symmetric function with $n$-bit input to a Boolean function with $\lceil \log (n+1) \rceil$-bit input. In the second step, we treat the $ \lceil \log (n+1) \rceil$ clean ancillary qubits as input qubits and the original $n$-qubit input qubits as borrowed ancillary qubits. By designing a parallelized quantum circuit with a depth of $O(\log^2 n)$, we are able to compute the output of the symmetric function on the target qubit.

We can further reduce the ancilla count to just one by incorporating an additional energy level (qutrit). Many physical systems inherently possess more than two energy levels, known as qudits \cite{ringbauer2022universal}. These higher energy levels can be harnessed to enhance information density \cite{zidac2023,bullock2005asymptotically} or serve as ancillary energy levels \cite{ISCAqutrit,chu2023scalable}. Gokhale et al. \cite{ISCAqutrit} utilized qutrits to implement a generalized $n$-Toffoli gate with a depth of $O(\log n)$ without ancillary qutrit. We extend their findings and implement arbitrary symmetric functions with a circuit depth of $O(\log^2 n)$ using only one ancillary qutrit in qutrit systems. This expansion broadens the range of achievable functions in qutrit systems from the $n$-Toffoli gate to any symmetric function, albeit with a slight increase in circuit depth and ancilla count.

The subsequent sections of this paper are structured as follows: Section~\ref{sec:pre} provides preliminaries. Sections~\ref{sec:qubit} and \ref{sec:qutrit} detail the implementation of symmetric functions using quantum circuits in qubit and qutrit systems, respectively. Section~\ref{sec:analysis} presents the analysis and evaluation of our results. Finally, Section~\ref{sec:conc} concludes the paper.

\section{Preliminary}
\label{sec:pre}

We define the symmetric function and quantum circuit implementation of a Boolean function as follows:

\begin{definition}
    A symmetric function $f: \{0,1\}^n \to \{0,1\}$ is a Boolean function that satisfies the following condition: for arbitrary permutation $\pi$ of the $n$ terms, we have $f(x_1,x_2,\dots,x_n) = f(x_{\pi(1)},x_{\pi(2)},\dots,x_{\pi(n)})$.
\end{definition}

Note that the output of a symmetric function is determined solely by the Hamming weight of its input.

\begin{definition}
    A quantum circuit $C$ implementing a Boolean function $f\!:\! \{0,1\}^n \! \to \! \{0,1\}$ means that $C\! \ket{\boldsymbol{x}}\! \ket{a}\! \ket{t} \! = \! \ket{\boldsymbol{x}}\! \ket{a}\! \ket{t \! \oplus \! f(\boldsymbol{x})}$, where $\boldsymbol{x} \in \{0,1\}^n$, $\ket{t}$ and $\ket{a}$ represents the target qubit and ancillary qubits. Circuit $C$ is composed of CNOT and single-qubit gates.
\end{definition}

The $Z(\theta)$ gate is defined by: $
        Z(\theta)= \begin{pmatrix}
            1 & 0 \\
            0 & e^{i\theta}
        \end{pmatrix}.$ 
        
The quantum fan-out gate $F_n$ is defined by: 
\begin{displaymath}
    F_n\ket{a}\ket{b_1}\ket{b_2}\cdots \ket{b_n} = \ket{a}\ket{b_1 \oplus a}\ket{b_2 \oplus a} \cdots \ket{b_n \oplus a}.
\end{displaymath}

The $Z(\theta)$ gate and quantum fan-out gate $F_n$ will be used in section~\ref{sec:qubit}. There is a well-known result regarding the quantum circuit implementation of $F_n$, which has been utilized in several papers \cite{takahashi2021power, fang2003quantum}. Therefore, we illustrate the circuit of $F_8$ in Fig.~\ref{fig:F_8} and omit the proof of Lemma~\ref{lem:fanout}.

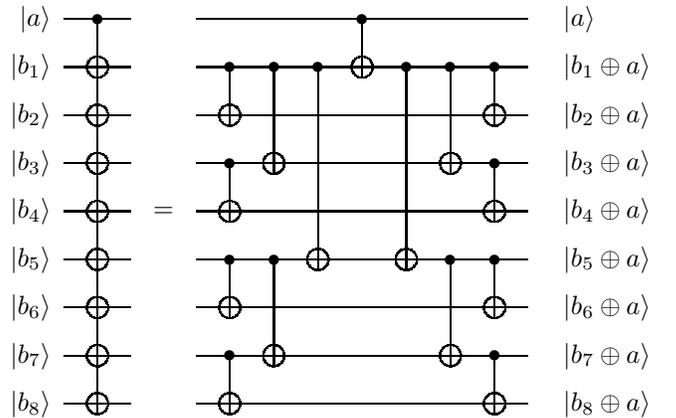
\begin{figure}[h]
\centering
\begin{flushleft}\ \ \ \ \ \ \ \ \ \Qcircuit @C=0.85em @R=1em @!R{
    \lstick{\ket{a}} & \ctrl{1} & \qw &&& \qw & \qw & \qw & \ctrl{1} & \qw & \qw & \qw & \qw & \rstick{\ket{a}} \\
    \lstick{\ket{b_1}} & \targ \qwx[1] & \qw &&& \ctrl{1} & \ctrl{2} & \ctrl{4} & \targ & \ctrl{4} & \ctrl{2} & \ctrl{1} & \qw & \rstick{\ket{b_1 \oplus a}} \\
    \lstick{\ket{b_2}} & \targ \qwx[1] & \qw &&& \targ & \qw & \qw & \qw & \qw & \qw & \targ & \qw & \rstick{\ket{b_2 \oplus a}} \\
    \lstick{\ket{b_3}} & \targ \qwx[1] & \qw &&& \ctrl{1} & \targ & \qw & \qw & \qw & \targ & \ctrl{1} & \qw & \rstick{\ket{b_3 \oplus a}} \\
    \lstick{\ket{b_4}} & \targ \qwx[1] & \qw &\push{=}&& \targ & \qw & \qw & \qw & \qw & \qw & \targ & \qw & \rstick{\ket{b_4 \oplus a}} \\
    \lstick{\ket{b_5}} & \targ \qwx[1] & \qw &&& \ctrl{1} & \ctrl{2} & \targ & \qw & \targ & \ctrl{2} & \ctrl{1} & \qw & \rstick{\ket{b_5 \oplus a}} \\
    \lstick{\ket{b_6}} & \targ \qwx[1] & \qw &&& \targ & \qw & \qw & \qw & \qw & \qw & \targ & \qw & \rstick{\ket{b_6 \oplus a}} \\
    \lstick{\ket{b_7}} & \targ \qwx[1] & \qw &&& \ctrl{1} & \targ & \qw & \qw & \qw & \targ & \ctrl{1} & \qw & \rstick{\ket{b_7 \oplus a}} \\
    \lstick{\ket{b_8}} & \targ & \qw &&& \targ & \qw & \qw & \qw & \qw & \qw & \targ & \qw & \rstick{\ket{b_8 \oplus a}}
}
\end{flushleft}
\caption{The quantum circuit implementing the $F_8$ gate.}
\label{fig:F_8}
\end{figure}

\begin{lemma}
\label{lem:fanout}
    The quantum fan-out gate $F_n$ can be implemented using a quantum circuit of depth $O(\log n)$ without ancillary qubit.
\end{lemma}

This paper introduces two types of ancillary qubits: clean and borrowed ancillary qubits. Clean ancillary qubits are promised to have a quantum state of $\ket{0}$ before use and it is required to recover them back to state $\ket{0}$ after use. Borrowed ancillary qubits can be in an arbitrary state before use and it is required that their quantum state remains unchanged after use.

A qutrit is a three-level quantum-mechanical system whose computation basis consists of states $\ket{0}$, $\ket{1}$, and $\ket{2}$. The definition of two types of ancillary qutrits is similar to that of qubits. When designing quantum circuits in qutrit systems, we use $X_{ij}$ gate and $\ket{0}$-$X_{01}$ gate as elementary gates. The $X_{ij}$ gate is defined as: $$X_{ij} = \ket{i}\bra{j}+\ket{j}\bra{i}+\ket{k}\bra{k},$$ where $k\neq i$ and $k\neq j$. The $\ket{0}$-$X_{01}$ gate applies the $X_{01}$ gate on the target qutrit if and only if the control qutrit state is $\ket{0}$. $X_{+1}$ gate and $X_{-1}$ gate will be utilized in Section~\ref{sec:qutrit}. Specifically, $X_{+1} = X_{02}X_{01} = \sum_i \ket{(i+1) \mod 3}\bra{i}$ and $X_{-1} = X_{+1}X_{+1}$.

\section{Efficient Implementation of Symmetric Functions}\label{sec:qubit}

We first introduce a quantum gate that can compute the Hamming weight of the input and store this value in the phase of a qubit.

\begin{definition}
    The function of quantum gate $C$-$M_n(\theta)$ is:
    \begin{equation*}
        \left( \alpha \ket{0}+\beta \ket{1} \right) \otimes \ket{\boldsymbol{x}} \to \left( \alpha \ket{0}+ \beta e^{i\theta(n-2|\boldsymbol{x}|)}\ket{1}\right) \otimes \ket{\boldsymbol{x}}.
    \end{equation*}
    Where $\theta \in \mathbb{R}, \ket{\boldsymbol{x}}=\ket{x_{1},x_{2},\cdots, x_{n}}, x_{i}\in \{0,1\}, \alpha \in \mathbb{Z}, b \in \mathbb{Z}, |\alpha|^2 + |\beta|^2 = 1$, and $|\boldsymbol{x}| = \sum_i x_i$.
\end{definition}

\begin{lemma}
\label{lem:CM}
    The $C$-$M_n(\theta)$ gate can be implemented using a quantum circuit of depth $O(\log n)$ without ancillary qubit.
\end{lemma}

\begin{proof}
The quantum circuit of $C$-$M_n(\theta)$ can be decomposed into three steps. Both Step 1 and Step 3 are fan-out gates on $n+1$ qubits, while Step 2 applies a $Z(\theta)$ gate to the last $n$ qubits. An example of $C$-$M_4(\theta)$ is depicted in Fig.~\ref{fig:M_n}. 

\begin{figure}[h]
    \begin{flushleft}
        \ \ \ \ \ \ \ \ \ \ \ \ \ \ \ \Qcircuit @C=0.2em @R=0.5em @!R {
        \lstick{\alpha \! \ket{0} \! + \! \beta \! \ket{1}} 
        & \push{\text{\circled{$\star$}}} \qw \qwx[1]
        & \qw &&& \ctrl{1} & \qw & \ctrl{1} & \qw & \rstick{\alpha \! \ket{0} \! + \! \beta e^{i\theta(4 - 2|\boldsymbol{x}|)} \! \ket{1}} \\
        \lstick{\ket{x_{1}}} & \multigate{3}{M_{4}(\theta)} & \qw &&& \targ \qwx[-1] & \gate{Z(\theta)} & \targ \qwx[-1] & \qw & \rstick{\ket{x_{1}}} \\
        \lstick{\ket{x_{2}}} & \ghost{M_{4}(\theta)} & \qw &\push{=}&& \targ \qwx[-1] & \gate{Z(\theta)} & \targ \qwx[-1] & \qw & \rstick{\ket{x_{2}}}\\
        \lstick{\ket{x_{3}}} & \ghost{M_{4}(\theta)} & \qw &&& \targ \qwx[-1] & \gate{Z(\theta)} & \targ \qwx[-1] & \qw & \rstick{\ket{x_{3}}}\\
        \lstick{\ket{x_{4}}} & \ghost{M_{4}(\theta)} & \qw &&& \targ \qwx[-1] & \gate{Z(\theta)} & \targ \qwx[-1] & \qw & \rstick{\ket{x_{4}}}\\
        }
    \end{flushleft}
    \caption{The quantum circuit of the $C$-$M_4(\theta)$ gate.}
    \label{fig:M_n}
\end{figure}
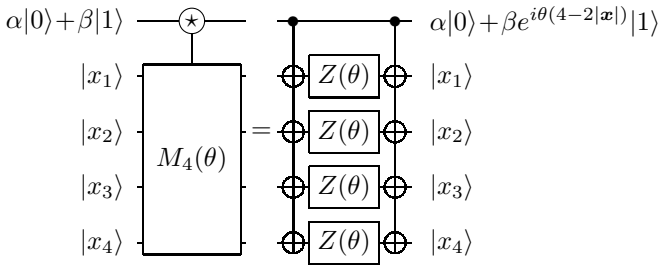

    For the correctness of this circuit, consider the initial state as follows: $(\alpha \ket{0}+\beta \ket{1}) \otimes \ket{x_1,\dots,x_n}.$
 After step 1, the quantum state is as follows:
    \begin{equation*}
        \alpha \ket{0}\ket{x_1,\dots,x_n} + \beta \ket{1}\ket{x_1\oplus 1,\dots,x_n \oplus 1}.
    \end{equation*}
    After step 2, the quantum state is as follows (the global phase $e^{i\theta |\boldsymbol{x}|}$ is omitted): $$
        \alpha \ket{0} \ket{x_1,\dots,x_n} + \beta e^{i\theta (n-2|\boldsymbol{x}|)} \ket{1} \ket{x_1 \oplus 1,\dots,x_n \oplus 1}.$$
    After the step 3, the quantum state is as follows: $$
        \left( \alpha \ket{0}+\beta e^{i\theta (n-2|\boldsymbol{x}|)}\ket{1} \right) \otimes \ket{x_1,\dots,x_n}.
    $$
    We can now ignore the global phase and conclude the proof of correctness. Note that the circuit depth is directly derived from Lemma~\ref{lem:fanout}. The global phase can also be eliminated by adding a fourth step to the circuit: applying one $Z(-\theta)$ gate on each of the last $n$ qubits, which only increases the circuit depth by one.
\end{proof}

\begin{figure*}[h]
    \centerline{
        \Qcircuit @C=1em @R=0.4em @!R {
        \lstick{\ket{0}_1} & \qw & \gate{H} & \gate{Z(\frac{2n\pi}{2^{m+1}})} & \qw & \qw & \qw &\push{\cdots} && \qw & \push{\text{\circled{$\star$}}} \qw \qwx[1] & \multigate{4}{QFT_{m}^{\dagger}}  & \qw &\rstick{\ket{b_1}} \\
        \lstick{\ket{0}_2} & \qw & \gate{H} & \gate{Z(\frac{2n\pi}{2^{m}})} & \qw & \qw & \qw &\push{\cdots} && \push{\text{\circled{$\star$}}} \qw \qwx[1] & \qw \qwx[1] & \ghost{QFT_{m}^{\dagger}} & \qw &\rstick{\ket{b_2}} \\
        \lstick{\vdots} &&\push{\vdots} &\push{\vdots}&&&&\push{\cdots}&& \qwx[1] & \qwx[1] &&&\\
        \lstick{\ket{0}_{m-1}} & \qw & \gate{H} & \gate{Z(\frac{2n\pi}{2^{3}})} & \qw & \push{\text{\circled{$\star$}}} \qw \qwx[1] & \qw &\push{\cdots} && \qw \qwx[1] & \qw \qwx[1] & \ghost{QFT_{m}^{\dagger}} &  \qw &\rstick{\ket{b_{m-1}}} \\
        \lstick{\ket{0}_{m}} & \qw & \gate{H} & \gate{Z(\frac{2n\pi}{2^{2}})} & \push{\text{\circled{$\star$}}} \qw \qwx[1] & \qw \qwx[1] & \qw &\push{\cdots} && \qw \qwx[1] & \qw \qwx[1] & \ghost{QFT_{m}^{\dagger}} &  \qw &\rstick{\ket{b_m}} \\
        \lstick{\ket{\boldsymbol{x}}} & {/} \qw & \qw & \qw & \gate{M_n(\frac{-2\pi}{2^2})} & \gate{M_n(\frac{-2\pi}{2^3})} & \qw &\push{\cdots} && \gate{M_n(\frac{-2\pi}{2^{m}})} & \gate{M_n(\frac{-2\pi}{2^{m+1}})} & \qw & \qw & \rstick{\ket{\boldsymbol{x}}}
        }}
    \caption{The quantum circuit of $C_{n}$. The input $\boldsymbol{x} \in \{0,1\}^n$, $m=\lceil \log (n+1) \rceil$. The binary representation of $|\boldsymbol{x}|$ is $b_1 b_2 \cdots b_{m-1}$.}
    \label{fig:Counting}
\end{figure*}
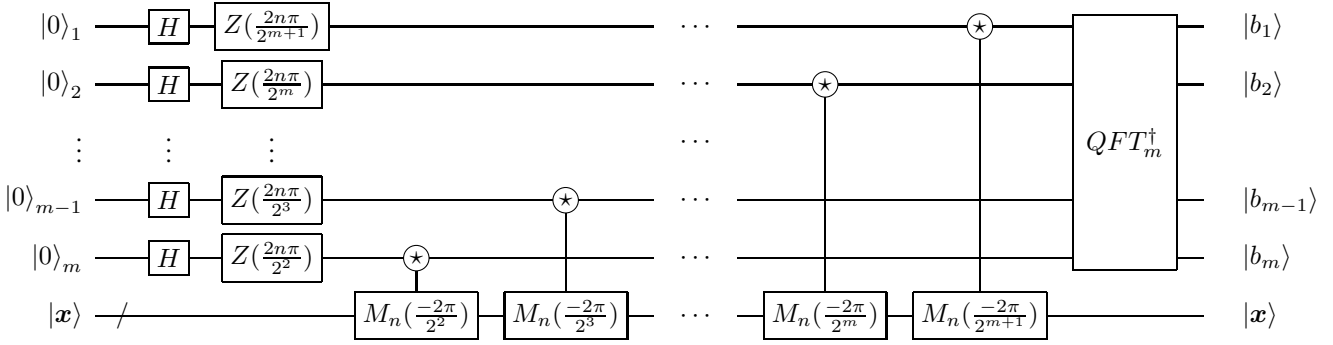

A similar idea as Lemma~\ref{lem:CM} is utilized for implementing logic-AND in \cite{gidney2018halving}.

Note that the symbol of the $C$-$M_n(\theta)$ gate depicts one control qubit and $n$ target qubits. However, in reality, only the relative phase of the "control qubit" is altered. This notation is used for convenience in representing quantum circuits. Next, we need to extract the value of $|\boldsymbol{x}|$ from the qubit phase onto computational basis. Note that $\lceil \log (n+1) \rceil$ output qubits that initialized to $\ket{0}$ are necessary for storing all the bits of $|\boldsymbol{x}|$.

\begin{definition}
 Define quantum circuit $C_{n}$ compute the Hamming weight: $C_{n} \ket{0}^{\otimes \lceil \log (n + 1) \rceil} \ket{\boldsymbol{x}} = \ket{|\boldsymbol{x}|}\ket{\boldsymbol{x}}$, where $\boldsymbol{x} \in \{0,1\}^n, |\boldsymbol{x}|=\sum_i x_i$. The first $\lceil \log (n + 1) \rceil$ qubits are output qubits.
\end{definition}

\begin{theorem}
\label{the:Cnm}
    The Hamming weight of $n$-bit string can be computed by quantum circuit $C_{n}$ of depth $O(\log^2 n)$ without ancillary qubits.
\end{theorem}

\begin{proof}
The quantum circuit for $C_{n}$ is shown in Fig.~\ref{fig:Counting}. The input $\boldsymbol{x} \in \{0,1\}^n$, and let $m = \lceil \log (n+1) \rceil$. The binary representation of $|\boldsymbol{x}|$ can be expressed by $|\boldsymbol{x}|_{(2)} = b_{1} b_{2} \cdots b_{m}$. 

To verify its correctness, we consider a specific output qubit $\ket{0}_j$ for $j=1,2,\dots,m$. Initially, we apply an $H$ gate and a $Z(\frac{2n\pi}{2^{m-j+2}})$ gate to this qubit. As a result, the quantum state is transformed to $\ket{0} \to \left(\ket{0} + e^{\frac{2in\pi}{2^{m-j+2}}} \ket{1} \right)/\sqrt{2}.$ 

Afterward, we apply a $C$-$M_n(\frac{-2\pi}{2^{m-j+2}})$ gate. According to Lemma~\ref{lem:CM}, the quantum state is transformed to:
\begin{align*}
    &\left(\ket{0} + e^{\frac{2in\pi}{2^{m-j+2}}-\frac{2i\pi}{2^{m-j+2}}(n-2|\boldsymbol{x}|)} \ket{1} \right)/\sqrt{2} \\
    =& \left(\ket{0} + e^{\frac{2i|\boldsymbol{x}|\pi}{2^{m-j+1}}}\ket{1} \right)/\sqrt{2}.
\end{align*}

Note that we have already encoded the information of $|\boldsymbol{x}|$ into the phase of $m$ output qubits. Next, we can use $QFT^{\dagger}_m$ to transpose $|\boldsymbol{x}|$ into computational basis. After applying $QFT^{\dagger}_{m}$ to the $m$ qubits, the resulting quantum state is $\ket{b_1}\ket{b_2}\cdots \ket{b_{m}}\ket{\boldsymbol{x}}=\ket{|\boldsymbol{x}|}\ket{\boldsymbol{x}}.$
Note that the depth of $QFT_m$ is $O(m)$ \cite{nielsen2010quantum}. Therefore, by Lemma~\ref{lem:CM}, the depth of the circuit in Fig.~\ref{fig:Counting} is $O(\log^2 n)$.
\end{proof}

    \begin{figure*}[h]
    \centerline{
        \Qcircuit @C=1em @R=0.2em @!R {
        \lstick{\ket{x_1}} & \ctrl{7} & \qw & \qw \barrier[-1.6em]{7} & \qw & \qw & \push{\cdots} && \qw & \multigate{6}{G_3} & \qw & \qw & \push{\cdots} && \qw & \multigate{6}{G_3^{\dagger}} & \qw \\
        \lstick{\ket{x_2}} & \qw & \ctrl{6} & \qw & \qw & \qw & \push{\cdots} && \qw & \ghost{G_3} & \qw & \qw & \push{\cdots} && \qw & \ghost{G_3^{\dagger}} & \qw \\
        \lstick{\ket{x_3}}  & \qw & \qw & \ctrl{5} & \qw & \qw & \push{\cdots} && \qw & \ghost{G_3} & \qw & \qw & \push{\cdots} && \qw & \ghost{G_3^{\dagger}} & \qw \\
        \lstick{\ket{a_1}}  & \qw & \qw & \qw & \ctrl{4} & \qw & \push{\cdots} && \qw & \ghost{G_3} & \qw & \qw & \push{\cdots} && \ctrl{4} & \ghost{G_3^{\dagger}} & \qw \\
        \lstick{\ket{a_2}}  & \qw & \qw & \qw & \qw & \qw & \push{\cdots} && \qw & \ghost{G_3} & \qw & \qw & \push{\cdots} && \qw & \ghost{G_3^{\dagger}} & \qw \\
        \lstick{\ket{a_3}}  & \qw & \qw & \qw & \qw & \qw & \push{\cdots} && \qw & \ghost{G_3} & \qw & \qw & \push{\cdots} && \qw & \ghost{G_3^{\dagger}} & \qw \\
        \lstick{\ket{a_4}}  & \qw & \qw & \qw & \qw & \qw & \push{\cdots} && \ctrl{1} & \ghost{G_3} & \ctrl{1} & \qw & \push{\cdots} && \qw & \ghost{G_3^{\dagger}} & \qw \\
        \lstick{\ket{t}}  & \gate{X^{c_1}}  & \gate{X^{c_2}} & \gate{X^{c_4}} & \gate{X^{c_3}} & \qw & \push{\cdots} && \gate{X^{c_7}} & \qw & \gate{X^{c_7}} & \qw & \push{\cdots} &&  \gate{X^{c_3}} & \gate{X^{c_0}} & \qw
        }}
    \caption{The circuit that implements a symmetric function $f:\{0,1\}^3 \to \{0,1\}$.}
    \label{fig:ESOP}
\end{figure*}
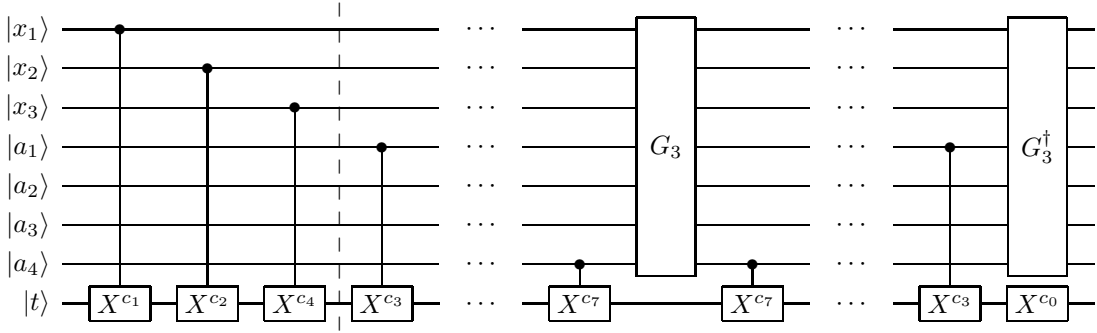
    
    \begin{figure}[h]
    \begin{flushleft}
        \ \ \ \ \ \ \ \ \ \ \ \ \ \ \ \Qcircuit @C=1.2em @R=0.4em @!R {
        \lstick{\ket{x_1}} & \qw & \ctrl{1} & \ctrl{2} & \qw & \qw & \qw & \rstick{\ket{x_1}}\\
        \lstick{\ket{x_2}} & \qw & \ctrl{2} & \qw & \ctrl{1} & \qw & \qw & \rstick{\ket{x_2}}\\
        \lstick{\ket{x_3}}  & \ctrl{1} & \qw & \ctrl{2} & \ctrl{3} & \ctrl{1} & \qw & \rstick{\ket{x_3}}\\
        \lstick{\ket{a_1}}  & \ctrl{3} & \targ & \qw & \qw & \ctrl{3} & \qw & \rstick{\ket{a_1 \oplus x_1 x_2}}\\
        \lstick{\ket{a_2}}  & \qw & \qw & \targ & \qw & \qw & \qw & \rstick{\ket{a_2 \oplus x_1 x_3}}\\
        \lstick{\ket{a_3}}  & \qw & \qw & \qw & \targ & \qw & \qw & \rstick{\ket{a_3 \oplus x_2 x_3}}\\
        \lstick{\ket{a_4}}  & \targ & \qw & \qw & \qw & \targ & \qw & \rstick{\ket{a_4 \oplus x_1 x_2 x_3}}
        }
    \end{flushleft}
    \caption{The quantum circuit of $G_3$.}
    \label{fig:G_3}
    \end{figure}
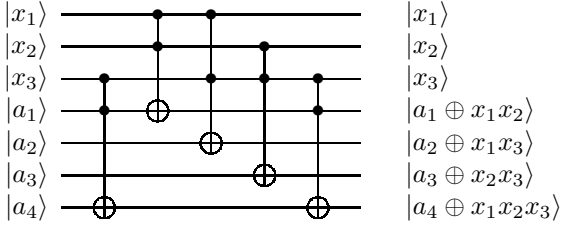

Our method for computing the Hamming weight can also be seen as based on the Phase Estimation Algorithm (PEA) \cite{nielsen2010quantum}. The utilization of PEA for computing the Hamming weight is well-known, as in previous research like \cite{wang2021preparing}. However, we are the first to offer a low-depth circuit implementation of PEA for computing the Hamming weight without ancillary qubits.

Theorem~\ref{the:Cnm} enables us to obtain the Hamming weight $|\boldsymbol{x}|$ from the input $\boldsymbol{x}$. For the majority function, we can easily implement by testing whether $|\boldsymbol{x}|\geq n/2$. For general symmetric functions, we need a new method to implement Boolean functions using borrowed ancillary qubits.

\begin{lemma}
    \label{lem:boolean_function}
    A Boolean function $f:\{0,1\}^n \to \{0,1\}$ can be implemented using a quantum circuit of depth $O(n^2)$ with $2^n - n - 1$ borrowed ancillary qubits.
\end{lemma}

\begin{proof}
    Assuming we can represent a Boolean function $f$ using its ESOP form \cite{sasao1993exmin2}: $f(\boldsymbol{x})=\bigoplus_{j=0}^{2^n-1} c_jf_j(\boldsymbol{x}), c_{j} \in \{0,1\}, f_j(\boldsymbol{x})=\bigwedge_{k=1}^n x_k^{b_k}, (j)_2=b_nb_{n-1}\cdots b_1.$

    There is a slight variation in the definition of the ESOP form used in this manuscript for the convenience of constructing quantum circuits. Note that the ESOP form is an XOR of multiple terms, and there are at most $2^n$ different terms. Our approach is to XOR each term of the ESOP form onto the target qubit $\ket{t}$, causing the target qubit $\ket{t}$ to become $\ket{t\oplus f(\boldsymbol{x})}$. 

    The quantum circuit that implements $f$ is divided into $3$ parts. The first part XORs all the linear terms onto the target qubit $\ket{t}$. For each $i \in \{1,2,\dots,n\}$, if $c_{2^i}=1$, we add a CNOT gate to the circuit with $x_i$ as control qubit and $\ket{t}$ as target qubit, indicating that $x_i$ is a term in the ESOP form of $f$. The second part requires $2^n-n-1$ borrowed ancillary qubits, where each ancillary qubit corresponds to a higher-order term. We add the gate $G_n$ to the circuit to XOR each higher-order term with its corresponding ancillary qubit. If $c_j=1$ for a higher-order term $f_j(\boldsymbol{x})$, we add a CNOT gate on both sides of $G_n$. The control qubit is the corresponding ancillary qubit for $f_j(\boldsymbol{x})$, and the target qubit is $\ket{t}$. Finally, we add $G_{n}^{\dagger}$ to restore all ancillary qubits. The third part adds an $X$ gate applied on $\ket{t}$ to the end of the circuit if $c_0=1$. This ensures correctness by noticing that all terms of the ESOP form have been XORed onto $\ket{t}$. Fig.~\ref{fig:ESOP} shows an example circuit for $n=3$. The Left of the dashed barrier corresponds to the first part. The circuit for $G_3$ is depicted in Fig.~\ref{fig:G_3}.
    
    Considering the implementation of $G_n$. The purpose of $G_n$ is to XOR all higher-order terms from the ESOP forms onto the corresponding borrowed ancillary qubits. Our approach to implementing $G_n$ is to start with $G_2$ and continuously use $G_k$ to implement $G_{k+1}$, eventually achieving $G_n$. For implementing $G_{k+1}$ based on $G_k$, we first use $k$ Toffoli gates to implement all quadratic terms containing $x_{k+1}$ (for example, the three Toffoli gates in the dashed box in Fig.~\ref{fig:G_4}). Then, we add $2^k - k - 1$ Toffoli gates before and after the circuit to implement all other terms. Fig.~\ref{fig:G_4} is an example circuit for implementing $G_4$. The correctness of the circuit comes from the fact that all monomials generated by $k+1$ variables comprise are merely monomials generated by $k$ variables along with these monomials multiplied by $x_{k+1}$.

    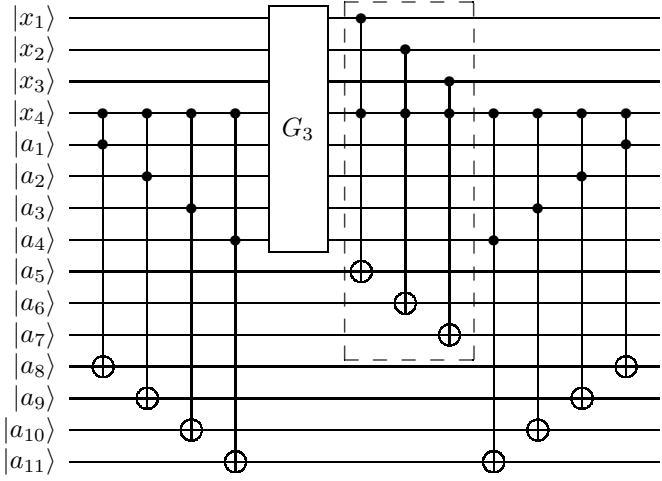
\begin{figure}[h]
    \begin{flushleft}
        \ \ \ \ \ \ \ \Qcircuit @C=0.85em @R=0.3em @!R {
        \lstick{\ket{x_1}} & \qw & \qw & \qw & \qw & \multigate{7}{G_3} & \ctrl{3} & \qw & \qw & \qw & \qw & \qw & \qw & \qw \\
        \lstick{\ket{x_2}} & \qw & \qw & \qw & \qw & \ghost{G_3} & \qw & \ctrl{2} & \qw & \qw & \qw & \qw & \qw & \qw \\
        \lstick{\ket{x_3}} & \qw & \qw & \qw & \qw & \ghost{G_3} & \qw & \qw & \ctrl{1} & \qw & \qw & \qw & \qw & \qw \\
        \lstick{\ket{x_4}} & \ctrl{1} & \ctrl{2} & \ctrl{3} & \ctrl{4} & \ghost{G_3} & \ctrl{5} & \ctrl{6} & \ctrl{7} & \ctrl{4} & \ctrl{3} & \ctrl{2} & \ctrl{1} & \qw \\
        \lstick{\ket{a_1}} & \ctrl{7} & \qw & \qw & \qw & \ghost{G_3} & \qw & \qw & \qw & \qw & \qw & \qw & \ctrl{7} & \qw \\
        \lstick{\ket{a_2}} & \qw & \ctrl{7} & \qw & \qw & \ghost{G_3} & \qw & \qw & \qw & \qw & \qw & \ctrl{7} & \qw & \qw \\
        \lstick{\ket{a_3}} & \qw & \qw & \ctrl{7} & \qw & \ghost{G_3} & \qw & \qw & \qw & \qw & \ctrl{7} & \qw & \qw & \qw \\
        \lstick{\ket{a_4}} & \qw & \qw & \qw & \ctrl{7} & \ghost{G_3} & \qw & \qw & \qw & \ctrl{7} & \qw & \qw & \qw & \qw \\
        \lstick{\ket{a_5}} & \qw & \qw & \qw & \qw & \qw & \targ & \qw & \qw & \qw & \qw & \qw & \qw & \qw \\
        \lstick{\ket{a_6}} & \qw & \qw & \qw & \qw & \qw & \qw & \targ & \qw & \qw & \qw & \qw & \qw & \qw \\
        \lstick{\ket{a_7}} & \qw & \qw & \qw & \qw & \qw & \qw & \qw & \targ & \qw & \qw & \qw & \qw & \qw \\
        \lstick{\ket{a_8}} & \targ & \qw & \qw & \qw & \qw & \qw & \qw & \qw & \qw & \qw & \qw & \targ & \qw \\
        \lstick{\ket{a_9}} & \qw & \targ & \qw & \qw & \qw & \qw & \qw & \qw & \qw & \qw & \targ & \qw & \qw \\
        \lstick{\ket{a_{10}}} & \qw & \qw & \targ & \qw & \qw & \qw & \qw & \qw & \qw & \targ & \qw & \qw & \qw \\
        \lstick{\ket{a_{11}}} & \qw & \qw & \qw & \targ & \qw & \qw & \qw & \qw & \targ & \qw & \qw & \qw & \qw \gategroup{1}{7}{11}{9}{1em}{--}
        }
    \end{flushleft}
    \caption{The quantum circuit of $G_4$.}
    \label{fig:G_4}
    \end{figure}
    
    To analyze circuit depth, we focus on the quantum gates on the left of the $G_n$ gate. There will be at most $2^n-1$ CNOT gates, acting on the target qubit $\ket{t}$ with distinct control qubits. Such a circuit can be implemented using a quantum parity gate with $O(n)$ circuit depth \cite{fang2003quantum}. The parallel implementation of quantum gates between $G_n$ and $G_{n}^{\dagger}$ follows a similar approach.

    For circuit depth of $G_n$, note that when constructing the circuit for $G_{k+1}$, we add $2^k-1$ Toffoli gates on the right side of the circuit for $G_k$. These Toffoli gates share a common control qubit $x_{k+1}$, while the other control qubits and target qubits are all distinct. Such a circuit can be implemented in parallel using quantum fan-out gates, as shown in Fig.~\ref{fig:Toffoli_fanout}. This shared-control Toffoli gate technique was proposed by Gokhale et al. \cite{gokhale2021quantum}. Using this technique to implement $2^k-1$ Toffoli gates, the circuit depth is primarily contributed by the four quantum fan-out gates $F_{2^k-1}$. According to Lemma~\ref{lem:fanout}, $2^k-1$ Toffoli gates only require a circuit depth of $O(k)$ for implementation. Similarly, the same method can be used for the $2^k-k-1$ Toffoli gates added on the left side of $G_k$. Therefore, $G_n$ can be implemented with a quantum circuit of depth $O(n^2)$.
\end{proof}

    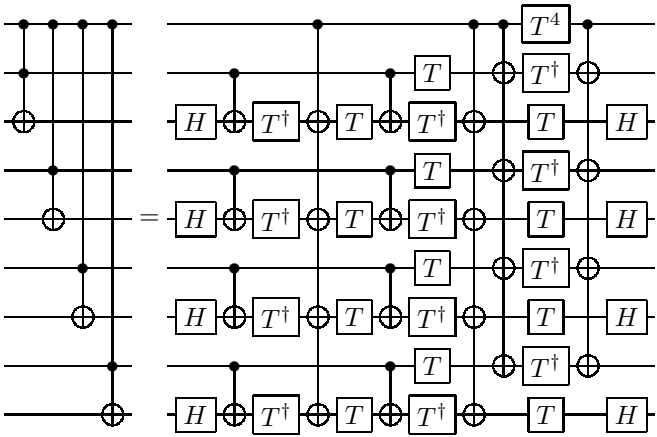
\begin{figure}[h]
    \centering
    \begin{flushleft}
        \Qcircuit @C=0.3em @R=0.4em @!R{
        & \ctrl{1} & \ctrl{3} & \ctrl{5} & \ctrl{7} & \qw &&& \qw & \qw & \qw & \ctrl{2} & \qw & \qw & \qw & \ctrl{2} & \ctrl{1} & \gate{T^4} & \ctrl{1} & \qw & \qw \\
        & \ctrl{1} & \qw & \qw & \qw & \qw &&& \qw & \ctrl{1} & \qw & \qw & \qw & \ctrl{1} & \gate{T} & \qw & \targ \qwx[2] & \gate{T^{\dagger}} & \targ \qwx[2] & \qw & \qw \\
        & \targ & \qw & \qw & \qw & \qw &&& \gate{H} & \targ & \gate{T^{\dagger}} & \targ \qwx[2] & \gate{T} & \targ & \gate{T^{\dagger}} & \targ \qwx[2] & \qw & \gate{T} & \qw & \gate{H} & \qw \\
        & \qw & \ctrl{1} & \qw & \qw & \qw &&& \qw & \ctrl{1} & \qw & \qw & \qw & \ctrl{1} & \gate{T} & \qw & \targ \qwx[2] & \gate{T^{\dagger}} & \targ \qwx[2] & \qw & \qw \\
        & \qw & \targ & \qw & \qw & \qw &\push{=}&& \gate{H} & \targ & \gate{T^{\dagger}} & \targ \qwx[2] & \gate{T} & \targ & \gate{T^{\dagger}} & \targ \qwx[2] & \qw & \gate{T} & \qw & \gate{H} & \qw \\
        & \qw & \qw & \ctrl{1} & \qw & \qw &&& \qw & \ctrl{1} & \qw & \qw & \qw & \ctrl{1} &\gate{T} & \qw & \targ \qwx[2] & \gate{T^{\dagger}} & \targ \qwx[2] & \qw & \qw  \\
        & \qw & \qw & \targ & \qw & \qw &&& \gate{H} & \targ & \gate{T^{\dagger}} & \targ \qwx[2] & \gate{T} & \targ & \gate{T^{\dagger}} & \targ \qwx[2] & \qw & \gate{T} & \qw & \gate{H} & \qw \\
        & \qw & \qw & \qw & \ctrl{1} & \qw &&& \qw & \ctrl{1} & \qw & \qw & \qw & \ctrl{1} & \gate{T} & \qw & \targ & \gate{T^{\dagger}} & \targ & \qw & \qw  \\
        & \qw & \qw & \qw & \targ & \qw &&& \gate{H} & \targ & \gate{T^{\dagger}} & \targ & \gate{T} & \targ & \gate{T^{\dagger}} & \targ & \qw & \gate{T} & \qw & \gate{H} & \qw
    }
    \end{flushleft}
    \caption{The quantum circuit of four shared-control Toffoli gates.}
    \label{fig:Toffoli_fanout}
    \end{figure}

    \begin{figure}[h]
    \begin{flushleft}
        \ \ \ \ \ \ \ \ \ \Qcircuit @C=1em @R=0.4em @!R {
        \lstick{\ket{x_{1:7}}} & {/} \qw & \multigate{5}{C_8} & \qw & \qw & \qw & \qw & \qw & \multigate{5}{C_8^{\dagger}} & \qw \\
        \lstick{\ket{x_8}} & \qw & \ghost{C_8} & \qw & \ctrl{1} & \multigate{4}{U_1} & \ctrl{1} & \multigate{4}{U_1} & \ghost{C_8^{\dagger}} & \qw \\
        \lstick{\ket{0}_1}  & \qw & \ghost{C_8} & \qw & \ctrl{4} & \ghost{U_1} & \ctrl{4} & \ghost{U_1} & \ghost{C_8^{\dagger}} & \qw \\
        \lstick{\ket{0}_2}  & \qw & \ghost{C_8} & \multigate{3}{U_0} & \qw & \ghost{U_1} & \qw & \ghost{U_1} & \ghost{C_8^{\dagger}} & \qw \\
        \lstick{\ket{0}_3}  & \qw & \ghost{C_8} & \ghost{U_0} & \qw & \ghost{U_1} & \qw & \ghost{U_1} & \ghost{C_8^{\dagger}} & \qw \\
        \lstick{\ket{0}_4}  & \qw & \ghost{C_8} & \ghost{U_0} & \qw & \ghost{U_1} & \qw & \ghost{U_1} & \ghost{C_8^{\dagger}} & \qw \\
        \lstick{\ket{t}}  & \qw & \qw & \ghost{U_0} & \targ & \qw & \targ & \qw & \qw & \qw 
        }
    \end{flushleft}
    \caption{The quantum circuit of a $8$-bit symmetric function.}
    \label{fig:symmetric}
    \end{figure}
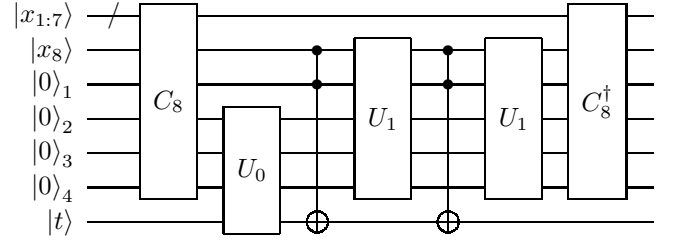


\begin{theorem}
    \label{the:qubit}
    A symmetric function $f(\boldsymbol{x})$ can be implemented using a quantum circuit of depth $O(\log^2 n)$ with $\lceil \log (n+1) \rceil$ clean ancillary qubits.
\end{theorem}

\begin{proof}
    Firstly, we construct the ESOP form of $g(|\boldsymbol{x}|)$, where $g(|\boldsymbol{x}|)=f(\boldsymbol{x})$. This step can be done by first constructing the truth table of $g(|\boldsymbol{x}|)$, and then transforming the truth table into the ESOP form mentioned in Lemma~\ref{lem:boolean_function}. The classical computation complexity can be bounded by $n^{O(1)}$ \cite{sasao1993exmin2}. Define $m=\lceil \log (n+1) \rceil$. We use $|\boldsymbol{x}|_1$ to denote the most significant bit of $|\boldsymbol{x}|$, and $|\boldsymbol{x}|_{2:m}$ to denote the remaining $m-1$ bits. The strategy is, we divide the terms of $g$ into two parts based on whether they contain $|\boldsymbol{x}|_1$, i.e. $g(|\boldsymbol{x}|) = g_0(|\boldsymbol{x}|_{2:m}) + |\boldsymbol{x}|_1 g_1(|\boldsymbol{x}|_{2:m})$ where $g_0$ and $g_1$ are Boolean functions that depends on $m-1$ variables. 

    We construct the quantum circuit as follows: First, we create a circuit that prepares $|x|$ on $m$ clean ancillary qubits, as described in Theorem~\ref{the:Cnm}. It is worth noting that by utilizing the $n$ input qubits as borrowed ancillary qubits, we can implement $g_0$ and $g_1$ through Lemma~\ref{lem:boolean_function}. The quantum circuit continues as follows: compute the result of $g_0(|\boldsymbol{x}|_{1:m-1})$ on $\ket{t}$ (the target qubit); add a Toffoli gate whose control qubits are $|\boldsymbol{x}|_m$ and a input qubit $\ket{a}$ that not used in previous step, target qubit is $\ket{t}$; compute the result of $g_1(|\boldsymbol{x}|_{1:m-1})$ on $\ket{a}$; add a Toffoli gate same as the previous step; compute the result of $g_1(|\boldsymbol{x}|_{1:m-1})$ on $\ket{a}$; recover the clean ancillary qubits. Our circuit first adds the result of $g_0(|\boldsymbol{x}|_{2:m})$ to $\ket{t}$. Then it uses a borrowed ancillary qubit $\ket{a}$ to add the result of $|\boldsymbol{x}|_0 g_1(|\boldsymbol{x}|_{2:m})$ to $\ket{t}$. Finally, we recover the state of $\ket{a}$ and $m$ clean ancillary qubits.

To see the correctness of our circuit, the circuit of $8$-bit symmetric function is depicted in Fig.~\ref{fig:symmetric}. First, $C_8$ compute $|\boldsymbol{x}| = \sum_{i=1}^8 x_i$ and store it in four clean ancillary qubits according to Theorem~\ref{the:Cnm}. Note that we have $f(\boldsymbol{x}) = g(|\boldsymbol{x}|) = g_0(|\boldsymbol{x}|_{2:4}) + |\boldsymbol{x}|_1 g_1(|\boldsymbol{x}|_{1:3})$. According to Lemma~\ref{lem:boolean_function}, we could construct $U_0$ to transform $\ket{t}$ to $\ket{t \oplus g_0(|\boldsymbol{x}|_{2:4})}$, and construct $U_1$ to transform $\ket{x_8}$ to $\ket{x_8 \oplus g_1(|\boldsymbol{x}|_{1:3})}$. Thus the transformation of target qubit $\ket{t}$ is as follows:
\begin{align*}
    \ket{t} & \to \ket{t \oplus g_0(|\boldsymbol{x}|_{2:4})}, \text{after $U_0$} \\
    & \to \ket{t \oplus g_0(|\boldsymbol{x}|_{2:4}) \oplus x_8 |\boldsymbol{x}|_1}, \text{after first Toffoli gate} \\
    & \to \ket{t \oplus g_0(|\boldsymbol{x}|_{2:4}) \oplus x_8 |\boldsymbol{x}|_1 \oplus (x_8 \oplus g_1(|\boldsymbol{x}|_{2:4}))|\boldsymbol{x}|_1} \\
    & = \ket{t \oplus g_0(|\boldsymbol{x}|_{2:4}) \oplus |\boldsymbol{x}|_1g_1(|\boldsymbol{x}|_{2:4})} \\
    & = \ket{t \oplus f(\boldsymbol{x})}, \text{after second Toffoli gate.}
\end{align*}

The aim of computing $f(\boldsymbol{x})$ by two parts is to obtain enough borrowed ancillary qubits to use Lemma~\ref{lem:boolean_function}. The circuit depth is $O(\log^2 n)$ by Theorem~\ref{the:Cnm} and Lemma~\ref{lem:boolean_function}.
    
\end{proof}

\section{Synthesis of symmetric functions using qutrits}
\label{sec:qutrit}

In this section, we raise a method for the synthesis of symmetric functions with the help of higher energy levels. Concretely, given any symmetric function $f$, we can implement $f$ in $O(\log^2 n)$ depth in qutrit systems using only $1$ clean ancillary qutrit. Note that the input and output of $f$ are still $0,1$-valued, and the extra energy level $\ket{2}$ is only for helping internal computing steps.

\subsection{Computing Hamming weight in qutrit systems}
The key idea of our design is to utilize the extra energy level to compute the Hamming weight in place, that is, the Hamming weight is placed on the input qutrits instead of the extra clean ancillary qutrit. As a result, we can compute Hamming weight by $O(\log n)$ depth circuit in qutrit systems with one ancillary qutrit, which may be of independent interest.

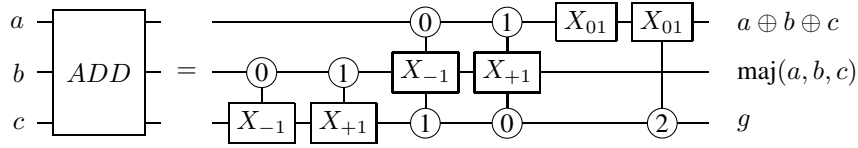
\begin{figure*}[h]
    \centerline{
        \Qcircuit @C=0.6em @R=0.4em @!R {
        \lstick{a} &\multigate{2}{ADD} &\qw &&& \qw & \qw & \push{\circled{0}} \qwx[1] \qw & \push{\circled{1}} \qwx[1] \qw & \gate{X_{01}} & \gate{X_{01}} \qwx[1] & \qw & \rstick{a\oplus b\oplus c}\\
        \lstick{b} &\ghost{ADD} &\qw &\push{=}&& \push{\circled{0}} \qwx[1] \qw & \push{\circled{1}} \qwx[1] \qw & \gate{X_{-1}} \qwx[1] & \gate{X_{+1}} \qwx[1] & \qw & \qw \qwx[1] & \qw &\rstick{\text{maj}(a,b,c)}\\
        \lstick{c} &\ghost{ADD} &\qw &&& \gate{X_{-1}} & \gate{X_{+1}} & \push{\circled{1}} \qw & \push{\circled{0}} \qw & \qw & \push{\circled{2}} \qw & \qw &\rstick{g}\\
    }}
    \caption{Qutrit full adder. $g$ represents garbage information.}
    \label{fig:qutrit_full_adder}
\end{figure*}

We define sets $Q_0, Q_1, \dots, Q_t$, where $t = \lceil \log (n+1) \rceil - 1$. The $n$ input qutrits are indexed by numbers in the set $\{1, 2, \dots, n\}$. An input qutrit only belongs to one set $Q_i$. The $k$-th input qutrit belonging to $Q_i$ is denoted by $k \in Q_i$. We want a qutrit in $Q_i$ to contribute $2^i$ to the total sum if its state is $\ket{1}$. Assuming the quantum state of the $k$-th qutrit is $\ket{y_k}$, our desired equation is: $|\boldsymbol{x}| = \sum_{i=0}^t 2^i \sum_{k\in Q_i}[y_k = 1]$, where $\boldsymbol{x} \in \{0,1\}^n$ is the input of symmetric function, $[y_k=1]=1$ if and only if $y_k=1$ otherwise $[y_k=1]=0$. Initially, $Q_0 = \{1, 2, \dots, n\}$ and $Q_i = \emptyset$ for $i \in \{1, 2, \dots, t\}$. If we can ensure $|Q_i| = 1$ for every $i \in \{0, 1, \dots, t\}$, the $t+1$ qutrits belonging to these sets will encode $|\boldsymbol{x}|$.

We design a qutrit full adder with inputs $a,b,c \in \{0,1\}$, represented by three qutrits. The circuit produces the outputs $a \oplus b \oplus c$ and $\text{majority}(a,b,c)$ as shown in Fig.~\ref{fig:qutrit_full_adder}. The circled number indicates the fire condition of the control qutrit. The qutrit full adder can be decomposed into a constant number of elementary gates \cite{zidac2023}, and its correctness can be verified through straightforward computation. It is worth noting that $a+b+c = a \oplus b \oplus c + 2 \cdot \text{majority}(a,b,c)$. This means that if we input three qutrits from $Q_i$ into the qutrit full adder, it will contribute one qutrit to $Q_i$ ($a \oplus b \oplus c$) and one qutrit to $Q_{i+1}$ ($\text{majority}(a,b,c)$). Therefore, we can repeatedly use the qutrit full adder to reduce the size of $|Q_i|$ if $|Q_i| \geq 3$. Algorithm \ref{alg:almost_counting} generate a quantum circuit to make $|Q_i| < 3$ for every set $Q_i$. This is the first phase of counting, denoted as the almost counting phase.

\begin{algorithm}
    \caption{Generating almost counting circuit in qutrit systems}\label{alg:almost_counting}
    \begin{algorithmic}[1]
        \State Let $C$ be an empty $n$-qutrit circuit whose qutrits are labelled by $\{1,2,\dots,n\}$
        \State $t\gets\lceil\log (n+1)\rceil-1, Q_0\gets \{1,2,\dots,n\},Q_1,\dots,Q_t\gets\emptyset$
        \While{$\exists i, |Q_i| \ge 3$} \label{algline:almost_counting_while}
            \State $M_0,\dots,M_{t-1}\gets\emptyset$ \Comment{Qubits in $M_i$ will be carries}
            \For{$i = 0,1,\dots,t-1$}
                \While{$|Q_i|\ge 3$}
                    \State Arbitrarily pick $a,b,c\in Q_i$
                    \State Add qutrit full adder on $a,b,c$ to $C$
                    \State $M_i\gets M_i\cup\{b\}$ \Comment{$b$ stores majority}
                    \State $Q_i\gets Q_i\setminus\{b,c\}$ \label{algline:almost_counting_reduction}
                \EndWhile
            \EndFor
            \For{$i=0,1,\dots, t-1$}
                \State $Q_{i+1}\gets Q_{i+1}\cup M_i$ \Comment{Add carries to the next level}
            \EndFor
        \EndWhile
    \end{algorithmic}
\end{algorithm}

\begin{lemma}\label{lem:almost_counting}
    The circuit $C$ generated by Algorithm \ref{alg:almost_counting} has the following properties: (1) With input $\ket{\boldsymbol{x}}$ in computational basis, the output of $C$ is pure state in computational basis denoted by $\ket{\boldsymbol{y}}$, and $|\boldsymbol{x}|=\sum_{i=0}^t2^i\sum_{k\in Q_i}[y_k=1]$. (2) There exists $0\le r\le t$ such that $1\le |Q_0|,\dots,|Q_{r-1}|\le 2$, and $|Q_r|,\dots,|Q_t|=0$.
\end{lemma}

\begin{proof}
    For the first condition, it is important to note that we only use qutrit full adder which is a reversible circuit. Therefore, if the input is a base state on computational basis then the output should be a base state on computational basis too. The equation in the first condition holds true at the beginning of the circuit. It can also be verified after applying any qutrit full adder. Regarding the second condition, we can make the following observations. Firstly, the qutrit full adder cannot reduce $|Q_i|$ to zero if $|Q_i| > 0$. Secondly, the circuit from Algorithm \ref{alg:almost_counting} can reduce $|Q_i|$ if $|Q_i| > 2$.
\end{proof}

The next step is the refining phase, where we aim to identify every $Q_i$ such that $|Q_i|=2$ and reduce its size to 1. In this phase, the quantum circuit resembles a ripple-carry adder \cite{cuccaro2004new}. If $|Q_i|=2$, we generate a qutrit in $Q_{i+1}$ from $Q_i$ until we produce the most significant bit of $|\boldsymbol{x}|$, which is stored in a clean ancillary qutrit. Then, we proceed to generate the remaining bits of $|\boldsymbol{x}|$.

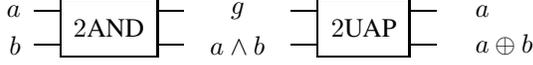
\begin{figure}[h]
    \centering
    \mbox{\Qcircuit @C=1em @R=0.4em @!R {
        \lstick{a} &\multigate{1}{2\text{AND}} &\qw &\push{g} &&\multigate{1}{2\text{UAP}} &\qw &\rstick{a}\\
        \lstick{b} &\ghost{2\text{AND}} &\qw &\push{a\wedge b} &&\ghost{2\text{UAP}} &\qw &\rstick{a\oplus b}\\
    }}
    \caption{$2$-bit AND gate and "Un-And and Parity" (2UAP) gate. $g$ represents garbage information.}
    \label{fig:2_and_uap_gate}
\end{figure}

    \begin{figure}[h]
    \centerline{
        \Qcircuit @C=1.2em @R=0.4em @!R {
        \lstick{\ket{a}} & \gate{X_{12}} & \gate{X_{+1}} \qwx[1] & \push{\circled{1}} \qwx[1] \qw & \qw & \rstick{\ket{g}}\\
        \lstick{\ket{b}} & \qw & \push{\circled{1}} \qw & \gate{X_{-1}} & \qw & \rstick{\ket{a \wedge b}}
        }}
    \caption{The quantum circuit of 2AND.}
    \label{fig:2AND}
    \end{figure}
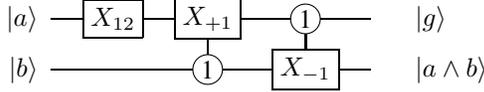

    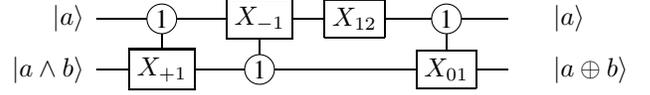
\begin{figure}[h]
    \centerline{
        \Qcircuit @C=1.2em @R=0.4em @!R {
        \lstick{\ket{a}} & \push{\circled{1}} \qwx[1] \qw & \gate{X_{-1}} \qwx[1] & \gate{X_{12}} & \push{\circled{1}} \qwx[1] \qw & \qw & \rstick{\ket{a}}\\
        \lstick{\ket{a \wedge b}} & \gate{X_{+1}} & \push{\circled{1}} \qw & \qw & \gate{X_{01}} & \qw & \rstick{\ket{a \oplus b}}
        }}
    \caption{The quantum circuit of 2UAP.}
    \label{fig:2UAP}
    \end{figure}

To generate a qutrit in $Q_{i+1}$ from two qutrits in $Q_{i}$, we utilize a $2$AND gate and a $2$UAP (Un-And and Parity) gate. The functionality of these gates is illustrated in Fig.~\ref{fig:2_and_uap_gate}. The $2$AND gate produces the logical AND of the two input bits, while the $2$UAP gate generates the XOR of the input bits. To implement the $2$AND gate, we consider the transformation $(a, b) \rightarrow (2a, b) \rightarrow (2a+b, b) \rightarrow (2a+b, b-[(2a+b)=1])$. It can be observed that $(b-[(2a+b)=1]) = a \wedge b$. Each step in this transformation can be implemented using elementary gates as shown in Fig.~\ref{fig:2AND}. The $2$UAP gate can be implemented by reversing the $2$AND gate and adding a $\ket{1}$-$X_{01}$ gate as shown in Fig.~\ref{fig:2UAP}.

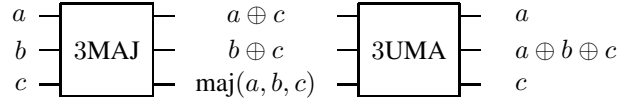
\begin{figure}[h]
    \begin{flushleft}
    \ \ \ \ \ \ \Qcircuit @C=0.8em @R=0.3em @!R {
        \lstick{a} &\multigate{2}{3\text{MAJ}} &\qw &\push{a\oplus c} &&\multigate{2}{3\text{UMA}} &\qw &\rstick{a}\\
        \lstick{b} &\ghost{3\text{MAJ}} &\qw &\push{b\oplus c} &&\ghost{3\text{UMA}} &\qw &\rstick{a\oplus b\oplus c}\\
        \lstick{c} &\ghost{3\text{MAJ}} &\qw &\push{\text{maj}(a,b,c)} &&\ghost{3\text{UMA}} &\qw &\rstick{c}\\
    }
    \end{flushleft}
    \caption{$3$-bit MAJority (3MAJ) and "Un-Majority and Add" (3UMA) gate in \cite{cuccaro2004new}.}
    \label{fig:3_maj_uma_gate}
\end{figure}

\begin{figure*}[h]
    \centering
    \mbox{\Qcircuit @C=1em @R=0.4em @!R {
        \lstick{\text{from} \ Q_0} &\multigate{1}{2\text{AND}} &\qw &\qw &\qw &\qw &\qw &\qw &\qw &\multigate{1}{2\text{UAP}} &\qw & \rstick{\text{in} \ Q_0}\\
        \lstick{\text{from} \ Q_0} &\ghost{2\text{AND}} &\multigate{1}{2\text{AND}} &\qw &\qw &\qw &\qw &\qw &\multigate{1}{2\text{UAP}} &\ghost{2\text{UAP}} &\qw\\
        \lstick{\text{from} \ Q_1} &\qw &\ghost{2\text{AND}} &\multigate{2}{3\text{MAJ}} &\qw &\qw &\qw &\multigate{2}{3\text{UMA}} &\ghost{2\text{UAP}} &\qw &\qw & \rstick{\text{in} \ Q_1}\\
        \lstick{\text{from} \ Q_2} &\qw &\qw &\ghost{3\text{MAJ}} &\qw &\qw &\qw &\ghost{3\text{UMA}} &\qw &\qw &\qw & \rstick{\text{in} \ Q_2}\\
        \lstick{\text{from} \ Q_2} &\qw &\qw &\ghost{3\text{MAJ}} &\multigate{1}{2\text{AND}} &\qw &\multigate{1}{2\text{UAP}} &\ghost{3\text{UMA}} &\qw &\qw &\qw \\
        \lstick{\text{from} \ Q_3} &\qw &\qw &\qw &\ghost{2\text{AND}} &\ctrl{1} &\ghost{2\text{UAP}} &\qw &\qw &\qw &\qw & \rstick{\text{in} \ Q_3}\\
        \lstick{\ket{0}} &\qw &\qw &\qw &\qw &\targ &\qw &\qw &\qw &\qw &\qw & \rstick{\text{in} \ Q_4}\\
    }}
    \caption{An example for the refining phase of counting circuit in qutrit systems.}
    \label{fig:refining_example}
\end{figure*}
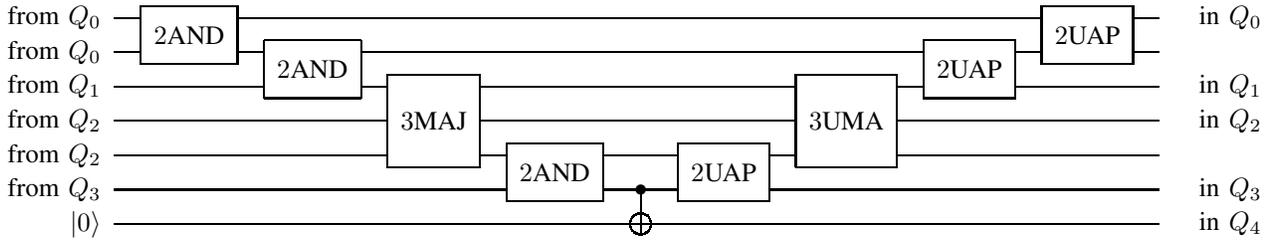

If $|Q_i|=3$ after applying the $2$AND gate to the qutrit in $Q_{i-1}$, we can utilize the $3$MAJ (Majority) gate and $3$UMA (Un-Majority and Add) gate \cite{cuccaro2004new}, illustrated in Fig.~\ref{fig:3_maj_uma_gate}, to achieve a similar functionality. $3$MAJ and $3$UMA have circuit implementations in qubit systems \cite{cuccaro2004new}. These circuits can be directly extended to qutrit systems and are therefore omitted. An example of the refining phase is illustrated in Fig.~\ref{fig:refining_example}.

\begin{theorem}
\label{the:qutrit_hamming}
    There is a quantum circuit of depth $O(\log n)$ in qutrit systems using one clean ancillary qutrit such that, with input $\ket{\boldsymbol{x},0}$ it computes Hamming weight $\ket{|\boldsymbol{x}|,\text{garbage}}$, where the most significant bit of $|\boldsymbol{x}|$ is stored in the ancillary qutrit.
\end{theorem}

\begin{proof}
The circuit consists of two phases: the almost counting phase and the refining phase, as described earlier. The correctness of the circuit follows directly from the discussion provided above. Note that in Algorithm~\ref{alg:almost_counting}, the gates generated in each iteration of the while loop (line 3) can be executed in parallel. For each layer of the qutrit full adder, there are no more than $2\lceil \log (n+1) \rceil$ qutrits that remain unchanged. Thus each layer of the qutrit full adder reduces the $\sum_i |Q_i|$ roughly by 1/3, resulting in the circuit depth generated by Algorithm~\ref{alg:almost_counting} is $O(\log n)$. The second phase also has a depth of $O(\log n)$ since there are $O(\log n)$ qutrits in all the sets $Q_i$ after the almost counting phase. This completes the proof.
\end{proof}

\subsection{Implementing Boolean functions in qutrit systems} 
This subsection shows that any $n$-bit Boolean function can be implemented by a quantum circuit of depth $O(n^2)$ using $2^n-n-1$ borrowed ancillary qutrit in qutrit systems. Our design is similar to the Boolean function implementation in Section \ref{sec:qubit}. The problem is that the borrowed ancillary qutrit we used is no longer suitable for operations over field $\mathbb{F}_2$. Luckily, we can still write any Boolean functions into polynomials in field $\mathbb{F}_3$.

    \begin{figure*}[h]
    \centerline{\Qcircuit @C=1em @R=0.4em @!R {
        \lstick{\ket{x_1}} & \qw & \qw & \push{\circled{1}} \qwx[1] \qw & \push{\circled{1}} \qwx[2] \qw & \qw & \qw & \qw & \qw & \rstick{\ket{x_1}}\\
        \lstick{\ket{x_2}} & \qw & \qw & \push{\circled{1}} \qwx[2] \qw & \qw & \push{\circled{1}} \qwx[1] \qw & \qw & \qw & \qw & \rstick{\ket{x_2}}\\
        \lstick{\ket{x_3}}  & \push{\circled{1}} \qwx[1] \qw & \push{\circled{1}} \qwx[1] \qw & \qw & \push{\circled{1}} \qwx[2] \qw & \push{\circled{1}} \qwx[3] \qw & \push{\circled{1}} \qwx[1] \qw & \push{\circled{1}} \qwx[1] \qw & \qw & \rstick{\ket{x_3}}\\
        \lstick{\ket{a_1}}  & \push{\circled{1}} \qwx[3] \qw & \push{\circled{2}} \qwx[3] \qw & \gate{X_{+1}} & \qw & \qw & \push{\circled{1}} \qwx[3] \qw & \push{\circled{2}} \qwx[3] \qw & \qw & \rstick{\ket{a_1 + x_1 x_2}}\\
        \lstick{\ket{a_2}}  & \qw & \qw & \qw & \gate{X_{+1}} & \qw & \qw & \qw & \qw & \rstick{\ket{a_2 + x_1 x_3}}\\
        \lstick{\ket{a_3}}  & \qw & \qw & \qw & \qw & \gate{X_{+1}} & \qw & \qw & \qw & \rstick{\ket{a_3 + x_2 x_3}}\\
        \lstick{\ket{a_4}}  & \gate{X_{-1}} & \gate{X_{+1}} & \qw & \qw & \qw & \gate{X_{+1}} & \gate{X_{-1}} & \qw & \rstick{\ket{a_4 + x_1 x_2 x_3}}
        }}
    \caption{The quantum circuit of $\mathbf{G}_3$ in qutrit.}
    \label{fig:qutrit_G_3}
    \end{figure*}
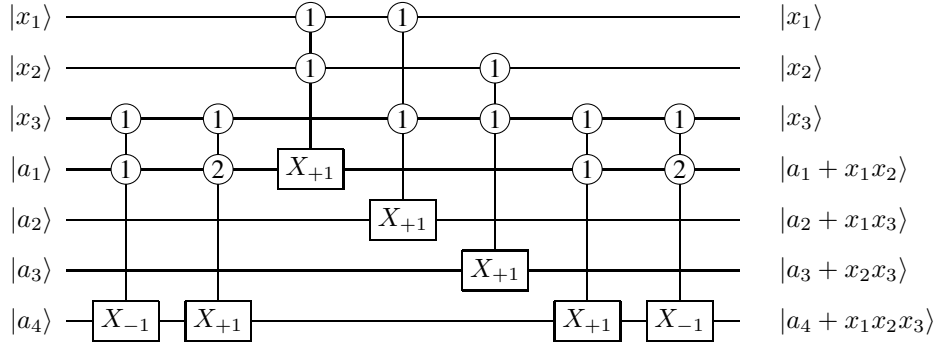

    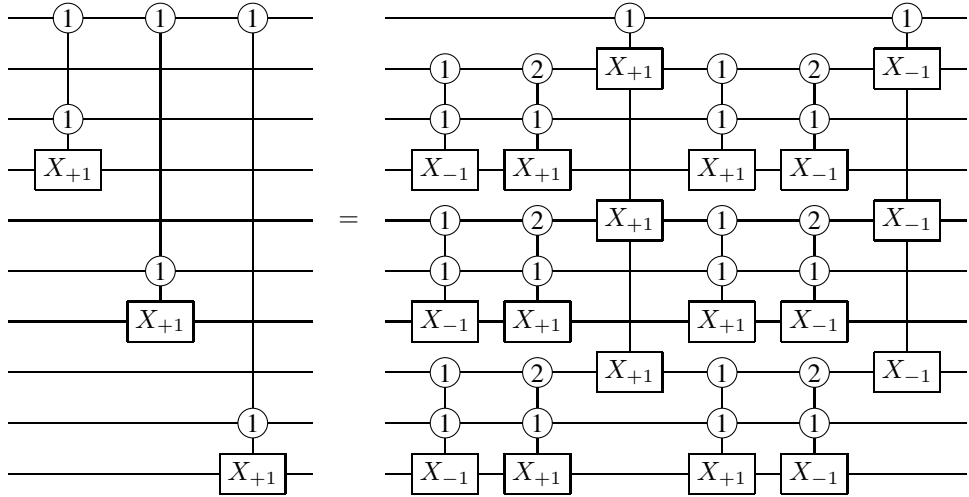
\begin{figure*}[h]
    \centerline{
        \Qcircuit @C=1em @R=0.4em @!R {
        & \push{\circled{1}} \qwx[2] \qw & \push{\circled{1}} \qwx[5] \qw & \push{\circled{1}} \qwx[8] \qw & \qw &&& \qw & \qw & \push{\circled{1}} \qwx[1] \qw & \qw & \qw & \push{\circled{1}} \qwx[1] \qw & \qw \\
        & \qw & \qw & \qw & \qw &&& \push{\circled{1}} \qwx[1] \qw & \push{\circled{2}} \qwx[1] \qw & \gate{X_{+1}} \qwx[3] & \push{\circled{1}} \qwx[1] \qw & \push{\circled{2}} \qwx[1] \qw & \gate{X_{-1}} \qwx[3] & \qw \\
        & \push{\circled{1}} \qwx[1] \qw & \qw & \qw & \qw &&& \push{\circled{1}} \qwx[1] \qw & \push{\circled{1}} \qwx[1] \qw & \qw & \push{\circled{1}} \qwx[1] \qw & \push{\circled{1}} \qwx[1] \qw & \qw & \qw \\
        & \gate{X_{+1}} & \qw & \qw & \qw &&& \gate{X_{-1}} & \gate{X_{+1}} & \qw & \gate{X_{+1}} & \gate{X_{-1}} & \qw & \qw \\
        & \qw & \qw & \qw & \qw &\push{=}&& \push{\circled{1}} \qwx[1] \qw & \push{\circled{2}} \qwx[1] \qw & \gate{X_{+1}} \qwx[3] & \push{\circled{1}} \qwx[1] \qw & \push{\circled{2}} \qwx[1] \qw & \gate{X_{-1}} \qwx[3] & \qw \\
        & \qw & \push{\circled{1}} \qwx[1] \qw & \qw & \qw &&& \push{\circled{1}} \qwx[1] \qw & \push{\circled{1}} \qwx[1] \qw & \qw & \push{\circled{1}} \qwx[1] \qw & \push{\circled{1}} \qwx[1] \qw & \qw & \qw \\
        & \qw & \gate{X_{+1}} & \qw & \qw &&& \gate{X_{-1}} & \gate{X_{+1}} & \qw & \gate{X_{+1}} & \gate{X_{-1}} & \qw & \qw \\
        & \qw & \qw & \qw & \qw &&& \push{\circled{1}} \qwx[1] \qw & \push{\circled{2}} \qwx[1] \qw & \gate{X_{+1}} & \push{\circled{1}} \qwx[1] \qw & \push{\circled{2}} \qwx[1] \qw & \gate{X_{-1}} & \qw \\
        & \qw & \qw & \push{\circled{1}} \qwx[1] \qw & \qw &&& \push{\circled{1}} \qwx[1] \qw & \push{\circled{1}} \qwx[1] \qw & \qw & \push{\circled{1}} \qwx[1] \qw & \push{\circled{1}} \qwx[1] \qw & \qw & \qw \\
        & \qw & \qw & \gate{X_{+1}} & \qw &&& \gate{X_{-1}} & \gate{X_{+1}} & \qw & \gate{X_{+1}} & \gate{X_{-1}} & \qw & \qw
        }}
    \caption{The quantum circuit of three shared-control qutrit gates.}
    \label{fig:qutrit_share}
    \end{figure*}

    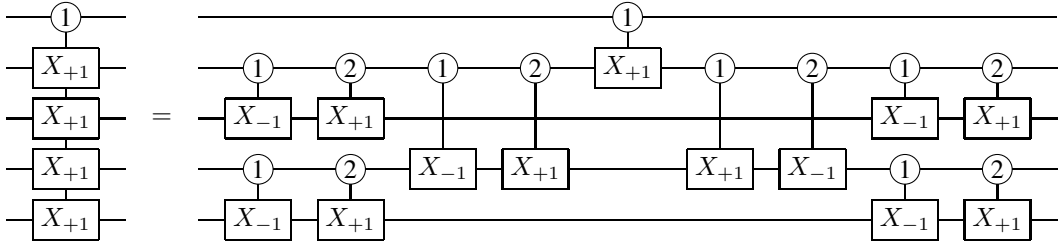
\begin{figure*}[h]
    \centerline{
        \Qcircuit @C=1em @R=0.4em @!R {
        & \push{\circled{1}} \qwx[1] \qw & \qw &&& \qw & \qw & \qw & \qw & \push{\circled{1}} \qwx[1] \qw & \qw & \qw & \qw & \qw & \qw \\
        & \gate{X_{+1}} \qwx[1] & \qw &&& \push{\circled{1}} \qwx[1] \qw & \push{\circled{2}} \qwx[1] \qw & \push{\circled{1}} \qwx[2] \qw & \push{\circled{2}} \qwx[2] \qw & \gate{X_{+1}} & \push{\circled{1}} \qwx[2] \qw & \push{\circled{2}} \qwx[2] \qw & \push{\circled{1}} \qwx[1] \qw & \push{\circled{2}} \qwx[1] \qw & \qw \\
        & \gate{X_{+1}} \qwx[1] & \qw &\push{=}&& \gate{X_{-1}} & \gate{X_{+1}} & \qw & \qw & \qw & \qw & \qw & \gate{X_{-1}} & \gate{X_{+1}} & \qw \\
        & \gate{X_{+1}} \qwx[1] & \qw &&& \push{\circled{1}} \qwx[1] \qw & \push{\circled{2}} \qwx[1] \qw & \gate{X_{-1}} & \gate{X_{+1}} & \qw & \gate{X_{+1}} & \gate{X_{-1}} & \push{\circled{1}} \qwx[1] \qw & \push{\circled{2}} \qwx[1] \qw & \qw \\
        & \gate{X_{+1}} & \qw &&& \gate{X_{-1}} & \gate{X_{+1}} & \qw & \qw & \qw & \qw & \qw & \gate{X_{-1}} & \gate{X_{+1}} & \qw
        }}
    \caption{The quantum circuit of quantum fan-out in qutrit systems.}
    \label{fig:qutrit_fanout}
    \end{figure*}

\begin{lemma}
\label{lem:F3poly}
    For any function $f:\{0,1\}^n\to\{0,1\}$, there are coefficients $c_S \in \{0,1,2\}, \ S\subseteq \{1,2,\dots,n\}$ such that $f(x_1,\dots,x_n)=\sum_{S\subseteq \{1,2,\dots,n\}}c_S\prod_{i\in S}x_i$, where the operations are over $\mathbb{F}_3$.
\end{lemma}

\begin{proof}
    Denote $g_0(x_1,\dots,x_{n-1}) = f(x_1,\dots,x_{n-1},0)$ and $g_1(x_1,\dots,x_{n-1}) = f(x_1,\dots,x_{n-1},1)$. We have: $f(x_1,\dots,x_n) = (1+2x_n)g_0(x_1,\dots,x_{n-1}) + x_n g_1(x_1,\dots,x_{n-1}).$
    
    By handling each input bit as $x_n$ above, we can obtain the form $f(x_1,\dots,x_n)=\sum_{S\subseteq \{1,2,\dots,n\}}c_S\prod_{i\in S}x_i$.
\end{proof}

\begin{lemma}
    \label{lem:boolean_qutrit}
    A Boolean function $f:\{0,1\}^n \to \{0,1\}$ can be implemented using a quantum circuit of depth $O(n^2)$ with $2^n - n - 1$ borrowed ancillary qutrit in qutrit systems.
\end{lemma}

\begin{proof}
    The quantum circuit has the same structure as the circuit constructed in Lemma~\ref{lem:boolean_function}. According to Lemma~\ref{lem:F3poly}, for the function $f(x_1,\dots,x_n)=\sum_{S\subseteq {1,2,\dots,n}}c_S\prod_{i\in S}x_i$, we need to add every term in this function to the target qutrit to transform $\ket{t}$ to $\ket{t+f(\boldsymbol{x})}$.

    For the linear term $c_S\prod_{i\in S}x_i$ where $|S| = 1$, we only need to add $x_i$ or $2x_i$ to the target qutrit according to the coefficient. The transformation $\ket{x,t} \to \ket{x,t+x}$ can be implemented by two qutrit gates: $\ket{1}$-$X_{+1}$ and $\ket{2}$-$X_{-1}$. Applying this transformation twice results in $\ket{x,t} \to \ket{x,t+2x}$. Thus, we can add all linear terms to the target qutrit $\ket{t}$.

    For each higher-order term $T_S = c_S\prod_{i\in S}x_i$ where $|S| > 1$, we require a borrowed ancillary qutrit to add this term to the target qutrit $\ket{t}$. Assuming we can add term $T_S$ to a borrowed ancillary qutrit $\ket{a}$, then $T_S$ could be added to the target qutrit $\ket{t}$ as follows:
    \begin{align*}
        \ket{a,t} &\to \ket{a,t+2a} \\
        &\to \ket{a+T_S,t+2a} \\
        &\to \ket{a+T_S,t+2a+a+T_S} \\
        &= \ket{a+T_S,t+T_S} \\
        &\to \ket{a,t+T_S}.
    \end{align*}
    Therefore, to add all higher-order terms to the target qutrit, we need $m=2^n-n-1$ borrowed ancillary qutrits $\ket{a_1,a_2,\dots,a_m}$. First, we add all the ancillary qutrits to the target qutrit twice:
    \begin{equation*}
        \ket{a_1,a_2,\dots,a_m}\ket{t} \to \ket{a_1,a_2,\dots,a_m}\ket{t + \sum_{i=1}^m 2a_i}.
    \end{equation*}
    Next, we add all the higher-order terms to the ancillary qutrit and then add all the ancillary qutrits to the target qutrit, transforming the target qutrit to $\ket{t+\sum_{|S|>1} T_S}$. Finally, we recover all the ancillary qutrits, successfully adding all higher-order terms to the target qutrit. In practice, if $c_S = 0$, $T_S$ does not need to be added to the target qutrit. If $c_S = 2$, $2\prod_{i\in S}x_i$ is difficult to compute, so it is replaced by adding $\prod_{i\in S}x_i$ twice. Therefore, we assume $c_S = 1$ for all higher-order terms for convenience.

    For the constant term $c_{\emptyset}$, we apply the $X_{+c_{\emptyset}}$ gate to the target qutrit. Next, we introduce two key techniques to reduce the circuit depth.

    We demonstrate the transformation: $$\ket{a_1,\dots,a_m}\ket{t} \to \ket{a_1,\dots,a_m}\ket{t+\sum_i a_i},$$ can be implemented by an $O(\log m) = O(n)$ depth circuit. Notably, we can add $\ket{a_{2k}}$ to $\ket{a_{2k+1}}$ for all $k = 0,1,\dots$ using a constant layer of gates. This results in the need to add the value of $\lceil m/2 \rceil$ qutrits to the target qutrit only. Therefore, by employing a constant-depth circuit, we can halve the size of the problem, reducing the circuit depth to $O(\log m) = O(n)$.

    We illustrate how to add all higher-order terms $T_S = \prod_{i\in S}x_i, |S| > 1$ to their corresponding ancillary qutrits using a $O(\log^2 m) = O(n^2)$ depth circuit. We set all coefficients $c_S$ to $1$ for convenience, as terms with $c_S = 2$ can be computed by adding them twice to the target qutrit. Let $\mathbf{G}_n$ denote the circuit that adds all $n$-bit higher-order terms to the ancillary qutrits, similar to the notation in Lemma~\ref{lem:boolean_function}. The quantum circuit $\mathbf{G}_3$ is illustrated in Fig.~\ref{fig:qutrit_G_3}, which is similar to the circuit in Fig.~\ref{fig:G_3}. Let's consider how $\mathbf{G}_3$ adds $x_1x_2x_3$ to the ancillary qutrit:
    \begin{align*}
        \ket{x_3,a_1,a_4} &\to \ket{x_3,a_1,a_4 + 2a_1x_3} \\
        &\to \ket{x_3,a_1 + x_1x_2,a_4 + 2a_1x_3} \\
        &\to \ket{x_3,a_1 + x_1x_2,a_4 + 2a_1x_3 + (a_1 + x_1x_2)x_3} \\
        & = \ket{x_3,a_1 + x_1x_2,a_4 + x_1x_2x_3}.
    \end{align*}
    
    The above transformation illustrates how we can add $x_1x_2x_3$ to the ancillary qutrit if we can compute $x_1x_2$. The key idea is that we can utilize the value of $x_1x_2$ that is added to $\ket{a_1}$ to compute the value of $x_1x_2x_3$. Similarly, we can add an arbitrary $x_iP$ to an ancillary qutrit if we can compute $P$. Thus, we can extend the circuit of $G_4$ in Fig.~\ref{fig:G_4} to the qutrit systems to construct the circuit of $\mathbf{G}_4$ based on $\mathbf{G}_3$. Similarly, we can construct $\mathbf{G}_k$ based on $\mathbf{G}_{k-1}$ for arbitrary $k$. Note that the key technique of constructing $G_k$ from $G_{k-1}$ using an $O(k)$ depth circuit involves shared-control Toffoli gates, as shown in Fig.~\ref{fig:Toffoli_fanout}. Likewise, the key technique in the qutrit systems is the circuit of shared-control qutrit gates, an example of which is depicted in Fig.~\ref{fig:qutrit_share}.
    
    The key idea is to utilize borrowed ancillary qutrits to store the information of the shared-control qutrit. This necessitates a quantum fan-out gate in the qutrit systems to distribute the value of the shared-control qutrit to all the ancillary qutrits. The $m$-qutrit fan-out gate can be implemented by an $O(\log m)$ depth circuit, with a proof similar to that of Lemma~\ref{lem:fanout}, which we omit here. An illustration of a $5$-qutrit fan-out gate is provided in Fig.~\ref{fig:qutrit_fanout}. In Fig.~\ref{fig:qutrit_fanout}, if the control qutrit is in state $\ket{1}$, the target qutrits are applied $X_{+1}$ gate. Reversing the circuit yields the circuit for applying $X_{-1}$ to the target qutrit. In Fig.~\ref{fig:qutrit_share}, we only depict the implementation of shared-control $\ket{11}$-$X_{+1}$ gates. The circuits for $\ket{12}$-$X_{+1}$ gate, $\ket{11}$-$X_{-1}$ gate, and $\ket{12}$-$X_{-1}$ gate are similar and omitted. Therefore, $\mathbf{G}_k$ can be constructed from $\mathbf{G}_{k-1}$ using an $O(k)$ depth circuit, similar to Lemma~\ref{lem:boolean_function}. Consequently, $\mathbf{G}_n$ can be implemented by an $O(n^2)$ depth circuit.
    
    With all the techniques mentioned above, computing $f(\boldsymbol{x})$ involves four steps:
    \begin{itemize}
        \item Adding all linear terms $c_S\prod_{i\in S}x_i$ where $|S| = 1$ to the target qutrit using an $O(n)$ depth circuit.
        \item Adding higher-order terms $\prod_{i\in S}x_i$ where $|S| > 1$ and $c_S > 1$ to the target qutrit using an $O(n^2)$ depth circuit with $2^n-n-1$ borrowed ancillary qutrit.
        \item Adding higher-order terms $\prod_{i\in S}x_i$ where $|S| > 1$ and $c_S = 2$ to the target qutrit using an $O(n^2)$ depth circuit with $2^n-n-1$ borrowed ancillary qutrit.
        \item Adding the constant term $c_{\emptyset}$ to the target qutrit using the $X_{+c_{\emptyset}}$ gate.
    \end{itemize}
    
    As a result, the depth of the entire circuit is $O(n^2)$. To construct the circuit of $\mathbf{G}_n$, $2^n-n-1$ borrowed ancillary qutrits are required.
    
    \end{proof}

\begin{table*}[h]
    \centering
    \caption{This table displays the circuit depth required to compute the Hamming weight, majority function, and arbitrary symmetric functions with varying numbers of input qubits. Additionally, it indicates the number of ancillary qubits needed to compute symmetric functions.}
    \label{tab:qubit}
    \begin{tabular}{ccccccccc}
        \hline \hline
        Number of input qubits & 31 & 63 & 127 & 255 & 511 & 1023 & 2047 & 4095 \\
        \hline
        Depth of circuit computing Hamming weight & 138 & 190 & 250 & 318 & 394 & 478 & 570 & 670 \\
        Depth of circuit computing majority function & 277 & 381 & 501 & 637 & 789 & 957 & 1141 & 1341 \\
        Depth of circuit computing symmetric functions & $\leq$1038 & $\leq$1682 & $\leq$2438 & $\leq$3306 & $\leq$4286 & $\leq$5378 & $\leq$6582 & $\leq$7898 \\
        Ancilla count of circuit computing symmetric functions & 5 & 6 & 7 & 8 & 9 & 10 & 11 & 12 \\ \hline \hline
    \end{tabular}
\end{table*}

\begin{table*}[h]
    \centering
    \caption{This table shows the number of layers of qutrit full adders required to compute the Hamming weight and majority function with varying numbers of input qubits without ancillary qutrit.}
    \label{tab:qutrit}
    \begin{tabular}{ccccccccc}
        \hline \hline
        Number of input qutrits & 31 & 63 & 127 & 255 & 511 & 1023 & 2047 & 4095 \\
        \hline
        Number of layers of qutrit full-adders to compute Hamming weight & 7 & 9 & 10 & 12 & 14 & 16 & 18 & 20 \\
        Number of layers of qutrit full-adders to compute majority function & 14 & 18 & 20 & 24 & 28 & 32 & 36 & 40 \\ \hline \hline
    \end{tabular}
\end{table*}

With Lemma~\ref{lem:boolean_qutrit}, we can now state Theorem~\ref{the:qutrit}.

\begin{theorem}
    \label{the:qutrit}
    A symmetric function $f$ can be implemented using a quantum circuit of depth $O(\log^2 n)$ with one clean ancillary qutrit in qutrit systems.
\end{theorem}

\begin{proof}
    Similar to the proof of Theorem~\ref{the:qubit}, we first compute the Hamming weight of the input $|\boldsymbol{x}|$ using Theorem~\ref{the:qutrit_hamming}. Then, we add the output $f(\boldsymbol{x})$ to the target qutrit in two parts, depending on whether the most significant bit of $|\boldsymbol{x}|$ is $0$ or $1$. This division of the computation of $f(\boldsymbol{x})$ into two parts is done to ensure we have enough borrowed ancillary qutrits, as required by Lemma~\ref{lem:boolean_qutrit}.

    Specifically, let's assume $f(\boldsymbol{x}) = g(|\boldsymbol{x}|) = g_0 + |\boldsymbol{x}|_1 g_1$, where $|\boldsymbol{x}|_1$ represents the most significant bit of $|\boldsymbol{x}|$, and $g_0$ and $g_1$ are functions of $|\boldsymbol{x}|_{2:\lceil \log(n+1) \rceil}$. After computing the Hamming weight $|\boldsymbol{x}|$ using Theorem~\ref{the:qutrit_hamming}, the transformation for computing $f(\boldsymbol{x})$ to the target qutrit $\ket{t}$ proceeds as follows, where $\ket{a}$ is an input qutrit serving as borrowed ancillary qutrit:
    \begin{align*}
        &\ket{|\boldsymbol{x}|_1,a,t} \\
        &\to \ket{|\boldsymbol{x}|_1,a,t+g_0}, \text{Lemma~\ref{lem:boolean_qutrit}} \\
        &\to \ket{|\boldsymbol{x}|_1,a,t+g_0 + 2a|\boldsymbol{x}|_1} \\
        &\to \ket{|\boldsymbol{x}|_1,a + g_1,t+g_0 + 2a|\boldsymbol{x}|_1}, \text{Lemma~\ref{lem:boolean_qutrit}} \\
        &\to \ket{|\boldsymbol{x}|_1,a + g_1,t+g_0 + 2a|\boldsymbol{x}|_1 + (a + g_1)|\boldsymbol{x}|_1} \\
        &= \ket{|\boldsymbol{x}|_1,a + g_1,t+g_0 + g_1|\boldsymbol{x}|_1} \\
        &= \ket{|\boldsymbol{x}|_1,a + g_1,t+f(\boldsymbol{x})} \\
        &\to \ket{|\boldsymbol{x}|_1,a,t+f(\boldsymbol{x})}, \text{Recover $\ket{a}$}.
    \end{align*}
    
    Constructing the quantum circuit based on the transformation above is straightforward and similar to Theorem~\ref{the:qubit}, thus it is omitted. As a result, the circuit depth is $O(\log^2 n)$ according to Theorem~\ref{the:qutrit_hamming} and Lemma~\ref{lem:boolean_qutrit}. One clean ancillary qutrit is required in Theorem~\ref{the:qutrit_hamming}.
\end{proof}

\section{Analysis and Evaluation}
\label{sec:analysis}

In this section, we provide an analysis of the depth of the quantum circuit that implements the Hamming weight and symmetric functions in qubit systems. We also evaluate the depth of the circuit that computes the majority function. It's important to note that all circuits are composed of CNOT and single-qubit gates. We omit a detailed analysis and restate the following Lemma and Theorem:

\begin{lemma}
    (Lemma~\ref{lem:fanout} restate) The quantum fan-out gate $F_n$ can be implemented using a quantum circuit of depth $2\lceil \log n \rceil + 1$ without ancillary qubit.
\end{lemma}

\begin{theorem}
\label{the:Hamming_d}
    (Theorem~\ref{the:Cnm} restate) The Hamming weight of $n$-bit string can be computed by quantum circuit $C_{n}$ of depth $4\lceil \log n \rceil \lceil \log(n+1) \rceil + 8\lceil \log(n+1) \rceil - 2$ without ancillary qubits.
\end{theorem}
\begin{proof}
    Note that the depth of the $k$-qubit Quantum Fourier Transform (QFT) is $5k-4$, achieved by decomposing each controlled rotation gate into two CNOT gates and two single-qubit rotation gates.
\end{proof}

\begin{lemma}
    (Lemma~\ref{lem:boolean_function} restate) A Boolean function $f:\{0,1\}^n \to \{0,1\}$ can be implemented using a quantum circuit of depth no more than $16n^2 + 36n - 154$ with $2^n - n - 1$ borrowed ancillary qubits.
\end{lemma}
\begin{proof}
    Upon analyzing the circuit construction detailed in Lemma~\ref{lem:boolean_function}, we find that the depth of $k$ shared-control Toffoli gates is $8\lceil \log k \rceil + 12$, and the depth of $G_n$ is $8n^2 + 16n - 76$ for $n \geq 5$. Utilizing these analyses, we determine the depth of the circuit that computes the Boolean function $f$.
\end{proof}

\begin{theorem}
\label{the:qubit_d}
    (Theorem~\ref{the:qubit} restate) A symmetric function can be implemented using a quantum circuit of depth no more than $48(\lceil \log(n+1) \rceil)^2 + 8\lceil \log(n+1) \rceil \lceil \log n \rceil + 28\lceil \log(n+1) \rceil - 502$ with $\lceil \log (n+1) \rceil$ clean ancillary qubits.
\end{theorem}

Here, we consider the depth of the circuit constructed by our method for various numbers of input qubits, specifically when $n \in \{31,63,127,255,511,1023,2047,4095\}$. The depth of the circuit computing the Hamming weight can be directly obtained using Theorem~\ref{the:Hamming_d}.

The majority function is a Boolean function that outputs $1$ if and only if at least half of the input bits are $1$. In our setup, where the number of input qubits has the form $2^k-1$ for some $k$, the output of the majority function is the most significant bit of the Hamming weight of the input. Therefore, we can compute the majority function by computing the Hamming weight and copying the most significant bit to the output qubit.

For arbitrary symmetric functions, Theorem~\ref{the:qubit_d} provides an upper bound on the circuit depth, as the depth varies depending on the specific symmetric function. The evaluation results are presented in Table~\ref{tab:qubit}. The results demonstrate that we can compute the Hamming weight and majority function with extremely low depth, even when the number of input qubits is large.

In qutrit systems, it's worth noting that when the number of input qubits has the form $2^k-1$, the Hamming weight can be computed solely by Algorithm~\ref{alg:almost_counting} without ancillary qutrit. We evaluate the number of layers of the full adder (depicted in Fig.~\ref{fig:qutrit_full_adder}) needed to compute the Hamming weight and majority function with different numbers of input qubits. The results are shown in Table~\ref{tab:qutrit}.

\section{Conclusion}
\label{sec:conc}
We propose a novel approach to implementing arbitrary symmetric function using a quantum circuit with only $O(\log^2 n)$ depth and $\lceil \log (n+1) \rceil$ ancillary qubits. The key technique is an $O(\log^2 n)$ depth circuit design to compute Hamming weight without ancillary qubit. Remarkably, in qutrit systems, the ancilla count can be further reduced to 1.  Our results demonstrate that symmetric functions can be implemented by quantum circuits with both low depth and low ancilla count. Additionally, our results partially address the question of space-time tradeoffs in quantum circuits raised by Maslov et al. \cite{maslov2021quantum}.

For future research, we are curious about whether the circuit depth can be further reduced to $O(\log n)$ to achieve asymptotic optimality. We are also interested in exploring the implementation of other classes of Boolean functions using poly-logarithmic or polynomial depth circuits. Additionally, the fault-tolerant circuit implementation of symmetric functions is intriguing; a direct approach could involve utilizing Hamming weight phasing \cite{kivlichan2020improved}.

\bibliographystyle{IEEEtran}
\bibliography{ref}

\end{document}